\documentclass[journal]{IEEEtran}

\usepackage{amsmath,amsfonts,amssymb}
\usepackage{array}
\usepackage[caption=false,font=scriptsize,labelfont=sf,textfont=sf]{subfig}
\usepackage{textcomp}
\usepackage{stfloats}
\usepackage{url}
\usepackage{verbatim}
\usepackage{graphicx}
\usepackage{cite}
\usepackage{xcolor}
\usepackage{pifont}
\usepackage{amsthm}
\usepackage{mathrsfs}
\usepackage{booktabs}
\usepackage{cases}
\usepackage{comment}
\usepackage{empheq}
\usepackage{bm}
\usepackage{algorithm}
\usepackage{algorithmic}
\usepackage[hidelinks]{hyperref}
\hypersetup{
  pdftitle={An Embodied Multiagent Framework Based on Token Communications for Cooperative ISAC},
  pdfauthor={Jiahe Guo, Jun Du, Chunxiao Jiang, Jintao Wang}
}

\newtheoremstyle{twcstyle}
  {}        
  {}        
  {\normalfont}   
  {}
  {\bfseries}     
  {.}
  {0.5em}  
  {}
\theoremstyle{twcstyle}

\newtheorem{theorem}{Theorem}

\makeatletter
\renewcommand{\maketag@@@}[1]{\hbox{\m@th\normalsize\normalfont#1}}%
\makeatother

\IEEEoverridecommandlockouts

\def\BibTeX{{\rm B\kern-.05em{\sc i\kern-.025em b}\kern-.08em
    T\kern-.1667em\lower.7ex\hbox{E}\kern-.125emX}}

\begin{document}
\bstctlcite{IEEEauthorControl}
\title{An Embodied Multiagent Framework Based on Token Communications for Cooperative ISAC}



\author{Jiahe Guo,~\IEEEmembership{Graduate Student Member,~IEEE},
Jun Du,~\IEEEmembership{Senior Member,~IEEE},\\
Chunxiao Jiang,~\IEEEmembership{Fellow,~IEEE},
and~Jintao Wang,~\IEEEmembership{Fellow,~IEEE}
\thanks{Jiahe Guo, Jun Du, and Jintao Wang are with the Department
of Electronic Engineering, Tsinghua University, Beijing 100084,
China. Jun Du and Jintao Wang are also with the State Key
Laboratory of Space Network and Communications, Tsinghua
University, Beijing 100084, China (e-mail: guojiahe1020@163.com;
jundu@tsinghua.edu.cn; wangjintao@tsinghua.edu.cn).}
\thanks{Chunxiao Jiang is with Beijing National Research Center
for Information Science and Technology, Tsinghua University,
Beijing 100084, China, and also with the State Key Laboratory
of Space Network and Communications, Tsinghua University,
Beijing 100084, China (e-mail: jchx@tsinghua.edu.cn).}
}

\maketitle

\begin{abstract}

The emerging low-altitude economy demands unmanned aerial vehicle (UAV)-enabled integrated sensing and communication (ISAC) for reliable connectivity and environmental awareness.
In particular, embodied UAV agents offer a promising means of supporting autonomous operations through a closed loop linking perception, decision-making, and physical actions. 
However, each UAV has access only to local observations, and effective cooperation requires exchanging local states and intentions. Directly sharing such information can incur substantial signaling overhead and hinder timely coordination in dynamic environments.
To deal with this problem, this paper investigates a cooperative ISAC network of embodied UAV agents and formulates a joint token communications (TokCom) and physical control problem to minimize total propulsion energy subject to communication and sensing rate requirements.
Then, we propose a state--intent TokCom (SI-TokCom) framework driven by multi-agent embodied policy learning. 
Specifically, separate pretrained codebooks enable compact exchanges of local states and intentions, while UAV agents jointly learn to select and compose tokens and determine physical actions based on local observations and received tokens.
Simulation results show that SI-TokCom achieves 98.9\% and 99.1\% of the centralized baseline's communication and sensing rates, respectively. Compared with the local baseline, it improves the corresponding rates by 5.0\% and 43.6\%, respectively, with essentially unchanged propulsion energy. 
These results highlight the potential of TokCom for communication-efficient cooperation among embodied UAV agents in ISAC systems.

\end{abstract}
 
\begin{IEEEkeywords}
Embodied intelligence, integrated sensing and communication, reinforcement learning, token communications.
\end{IEEEkeywords}

\section{Introduction}
\label{sec:introduction}




Sixth-generation (6G) wireless networks are envisioned to support emerging applications that require both high-quality communication and accurate environmental sensing~\cite{letaief2021edge, du2020machine}. With the rapid growth of wireless services, the demand for spectrum and hardware resources for communication and sensing is continuously increasing. Since radar sensing and wireless communication systems exhibit significant commonality in spectrum usage and radio frequency (RF) front-end architectures, integrating these two functionalities has become a promising approach to improve spectral and hardware efficiency. This trend has led to the emergence of integrated sensing and communication (ISAC)~\cite{Liu2022ISAC}.
Meanwhile, the emerging low-altitude economy (LAE) is driving demand for embodied unmanned aerial vehicle (UAV) agents that integrate perception, autonomous decision-making, and action in a closed loop~\cite{wu2026toward}. 
Embodied agents can enhance the operational flexibility of ISAC, motivating cooperative multi-UAV ISAC, where agents coordinate their mobility and resource allocation to serve geographically distributed users and targets~\cite{xiang2026aligning}.


However, conventional optimization methods for multi-UAV ISAC often involve high computational complexity and rely on detailed prior knowledge of the system. Specifically, joint mobility control and resource allocation require frequent exchanges of local states and intended actions~\cite{Zhu2024CommSurvey}, as UAV movement directly affects wireless channels, sensing geometry, and propulsion energy consumption~\cite{Meng2024UAVISAC,Zeng2019UAVTutorial}. These exchanges can consume substantial communication resources and introduce transmission delays. Moreover, iterative optimization algorithms may incur additional computational latency~\cite{Meng2024UAVISAC}. The resulting delays can render shared information outdated and impair coordinated decision-making~\cite{Yates2021AoISurvey}. The challenge is therefore to achieve energy-efficient mobility control and resource allocation in real time under highly dynamic and stochastic conditions.

Recent advances in artificial intelligence (AI) and the emergence of embodied agents offer new opportunities for lightweight coordination in multi-UAV ISAC systems through learning-based control and semantic communication (SemCom). Instead of transmitting raw observations or full-dimensional features, SemCom exchanges task-relevant information to reduce signaling overhead~\cite{luo2022semantic}. However, practical implementations may still rely on complex models with substantial computational and memory costs.
Recently, token communications (TokCom) have emerged as a promising approach for conveying information through discrete semantic tokens associated with shared embeddings, allowing agents to exchange compact indices instead of full-dimensional representations~\cite{ma2026semantic,qiao2025token}. For cooperative UAV-enabled ISAC, such compact exchanges can reduce signaling overhead and transmission latency while providing decision-relevant information for coordinated mobility control and resource allocation~\cite{yi2026token}. TokCom can therefore serve as a communication enabler for embodied UAV agents, supporting timely cooperation within the perception-action loop~\cite{yang2026embodied,niu2024comprehensive}.


\subsection{Related Work}
\label{subsec:related_work}


\subsubsection{Cooperative ISAC}

Existing studies have investigated cooperative ISAC among multiple base stations (BSs)~\cite{wei2024integrated,yang2025cooperative} and extended this research to UAV-enabled networks. Iterative optimization approaches have been widely adopted for cooperative ISAC~\cite{xiu2024improving,Pan2024CooperativeISAC}. In particular, the authors in~\cite{Zhou2026JRC} developed a centralized alternating optimization framework for UAV deployment and resource allocation. Xiu et al.~\cite{xiu2025movable} proposed an iterative optimization algorithm based on block coordinate descent (BCD) and successive convex approximation (SCA) for ISAC.

Learning-based methods have also been explored for adaptive UAV coordination~\cite{zhang2025task, chen2026differential}. Wang et al.~\cite{wang2026multi} proposed hierarchical graph-based multi-agent reinforcement learning (MARL), combining global coordination at the BS with local decision-making by UAVs. A centralized soft actor--critic approach was developed for joint trajectory planning and resource allocation in~\cite{Qin2023DRLISAC}, while MARL was employed for distributed sensing and communication coordination in~\cite{Xie2024DistributedISAC}. Learning-based cooperation has further been extended to integrated sensing, communication, and computation~\cite{Liu2025ISCC}, including attention-assisted multi-agent proximal policy optimization (MAPPO) for cooperative trajectory and resource control~\cite{Peng2025ISCC}.

More recently, SemCom has been explored to reduce signaling overhead by conveying compact, task-relevant representations. Existing studies have investigated task-oriented information exchange in cooperative UAV networks~\cite{deng2025semantic}, including multimodal semantic transmission for relay cooperation~\cite{Hu2025SemanticUAV}. 
For UAV-assisted ISAC, Guo et al.~\cite{Guo2026PerceptionISAC} combined perception-enhanced multimodal semantic transmission with semantic-aware resource allocation.
Nevertheless, existing approaches may still incur substantial overhead from state exchange, iterative computation, or the transmission of continuous semantic representations. Efficient information exchange among UAVs under limited energy budgets remains insufficiently explored.


\subsubsection{TokCom}

Building on SemCom, TokCom uses tokens as common units for semantic representation, processing, and transmission~\cite{ma2026semantic}. An early framework~\cite{qiao2025token} exploited foundation models and cross-modal context for efficient semantic transmission. For compact token learning, Wei et al.~\cite{wei2025token} adopted an information bottleneck approach to preserve essential multimodal information. A geometric method~\cite{park2026geometric} selects informative tokens through cross-modal dependencies under latency constraints. The authors in~\cite{shin2026context} used a shared masked language model for joint token masking and detection, reducing redundant transmissions and enabling context-assisted recovery.


Research has further extended TokCom to wireless transmission and network coordination. Hierarchical TokCom~\cite{lee2026hierarchical} exploits semantic similarity to improve transmission robustness, while Ada-TokenCom~\cite{zhang2026ada} jointly adapts token compression and transmission schemes to channel conditions. Qiao et al.~\cite{qiao2025todma} proposed token-domain multiple access with context-assisted token recovery for concurrent uplink transmissions. For distributed multimodal sensing, AirTF~\cite{liu2026airtf} exploits wireless signal superposition to fuse spatially aligned tokens. At the system level, Jiang et al.~\cite{jiang2026agentic} introduced agentic TokCom to adapt communication workflows to task and resource requirements.
These studies focus on token representation and transmission efficiency. However, using TokCom to support cooperation among embodied agents remains insufficiently explored in the literature.


\subsection{Motivation and Contributions}
\label{subsec:motivation_contributions}
Despite these advances, efficient information exchange remains a practical challenge for cooperative UAV-enabled ISAC. Centralized coordination requires global state collection and control dissemination, incurring signaling overhead and delay. Although SemCom can reduce data volume, its representations and associated models still incur communication, computation, and memory costs. Recent advances in TokCom motivate its application to this setting, using compact tokens to represent UAV states and intended actions. However, the value of exchanged tokens depends on their influence on the recipients' subsequent decisions, which are shaped by local observations and physical constraints. This coupling calls for jointly optimizing token selection and composition with mobility control and resource allocation under limited communication payload and energy budgets.

To address these challenges, we propose a state--intent TokCom (SI-TokCom) framework that integrates offline semantic codebook pretraining with online token exchange and physical control. During execution, embodied UAV agents coordinate their decisions using local observations and received tokens, supported by analytical beamforming updates. The main contributions are summarized as follows.

\begin{itemize}

    \item We formulate a cooperative multi-UAV ISAC problem that couples token codebook design and selection with UAV trajectories, user association, bandwidth allocation, transmit power, and beamforming to minimize total propulsion energy under communication and sensing rate requirements. We derive the position Fisher information matrix (FIM) and the corresponding squared position error bound (SPEB) for cooperative target localization. The sequential decision process is modeled as a decentralized partially observable Markov decision process (Dec-POMDP).

    \item We develop separate shared state token (S-Token) and intent token (I-Token) codebooks to represent current operating states and intended behavior, respectively. Offline pretraining combines vector quantization (VQ), semantic reconstruction, expert-action consistency, and codeword-usage regularization to retain decision-relevant information. The codebooks and semantic decoders are subsequently frozen, preserving shared semantic mappings while token selection and composition adapt through online policy learning.

    \item We design an online alternating optimization algorithm using MAPPO for token selection and composition followed by mixed discrete and continuous physical decisions, with closed-form beamforming updates. Counterfactual token masking assigns token-specific credit by evaluating changes in the team reward. Closed-form beamforming updates are derived by applying the Karush--Kuhn--Tucker (KKT) conditions to the subproblems reformulated via fractional programming (FP). This decomposition removes complex-valued beamforming matrices from the learned action space.
    
    \item We perform numerical simulations to evaluate the performance of the proposed SI-TokCom framework. In the evaluated setting, SI-TokCom uses a coordination payload of 96 bits per UAV per time slot and achieves 98.9\% and 99.1\% of the centralized baseline's communication and sensing rates, respectively. It also improves the communication and sensing rates by 5.0\% and 43.6\%, respectively, over the baseline with essentially unchanged propulsion energy. In addition, we evaluate the SPEB to characterize localization performance. Token ablations, trajectory examples, and codebook-usage comparisons further illustrate the learned coordination behavior.

\end{itemize}
\subsection{Organization and Notation}

The remainder of this paper is organized as follows. Section~II presents the system model and problem formulation. Section~III develops the proposed SI-TokCom framework, including offline codebook pretraining and online alternating optimization. Section~IV presents the simulation results, and Section~V concludes the paper.

\textit{Notation:} $(\cdot)^{T}$, $(\cdot)^H$, and $(\cdot)^{-1}$ denote transpose, conjugate transpose, and inverse, respectively. $|\cdot|$ denotes the absolute value or modulus of a scalar. $[\mathbf A]_{i,j}$ denotes the $(i,j)$th entry of matrix $\mathbf A$. $\mathbf A^\star$ denotes an optimal solution for the matrix variable $\mathbf A$. $\|\cdot\|_2$ and $\|\cdot\|_F$ denote the Euclidean and Frobenius norms, respectively. $\mathbb{E}\{\cdot\}$ denotes expectation. $\operatorname{Im}\{\cdot\}$ denotes the imaginary part of a complex scalar or matrix. $\operatorname{tr}(\cdot)$ denotes the matrix trace, and $\operatorname{diag}(\cdot)$ constructs a diagonal matrix from its arguments. $\mathbf I_n$ denotes the $n\times n$ identity matrix, and $[x]^+\triangleq\max\{x,0\}$. $\operatorname{sg}[\cdot]$ denotes the stop-gradient operator. $D_{\rm KL}(\cdot\|\cdot)$ denotes the Kullback--Leibler (KL) divergence between two probability distributions.

\section{System Model and Problem Formulation}

As illustrated in Fig.~\ref{fig:system_model}, we consider a cooperative multi-UAV ISAC system comprising \(K\) UAVs, \(N_c\) ground communication users (CUs), and \(N_s\) ground sensing targets (STs) over \(N_T\) time slots. The UAVs cooperatively provide data services to the CUs and perform detailed sensing of the STs based on coarse prior estimates of their positions~\cite{khalili2024efficient}.
Their index sets are denoted by \(\mathcal{K}\triangleq\{1,\ldots,K\}\), \(\mathcal{C}\triangleq\{1,\ldots,N_c\}\), \(\mathcal{S}\triangleq\{1,\ldots,N_s\}\), and \(\mathcal{T}\triangleq\{1,\ldots,N_T\}\), respectively. Each UAV is equipped with an \(M\)-element transmit antenna array and a single receive antenna, while each CU has a single receive antenna. The position and velocity of UAV \(k\) in time slot \(t\) are denoted by \(\mathbf{q}_k[t]\triangleq[x_k[t],y_k[t],H]^{T}\in\mathbb{R}^{3}\) and \(\mathbf{v}_k[t]\in\mathbb{R}^{3}\), respectively, with \(\mathbf{q}_k[t+1]=\mathbf{q}_k[t]+T\mathbf{v}_k[t]\), where \(t\) is the slot index and \(T\) is the duration of each slot. All UAVs fly at the same constant altitude \(H\). The fixed ground positions of CU \(i\) and ST \(j\) are denoted by \(\mathbf{p}_i^{c}\triangleq[x_i^{c},y_i^{c},0]^{T}\) and \(\mathbf{p}_j^{s}\triangleq[x_j^{s},y_j^{s},0]^{T}\), respectively.

\begin{figure}
    \centering
    \includegraphics[
        width=\linewidth,
        trim={0mm 4cm 0mm 0mm},
        clip
    ]{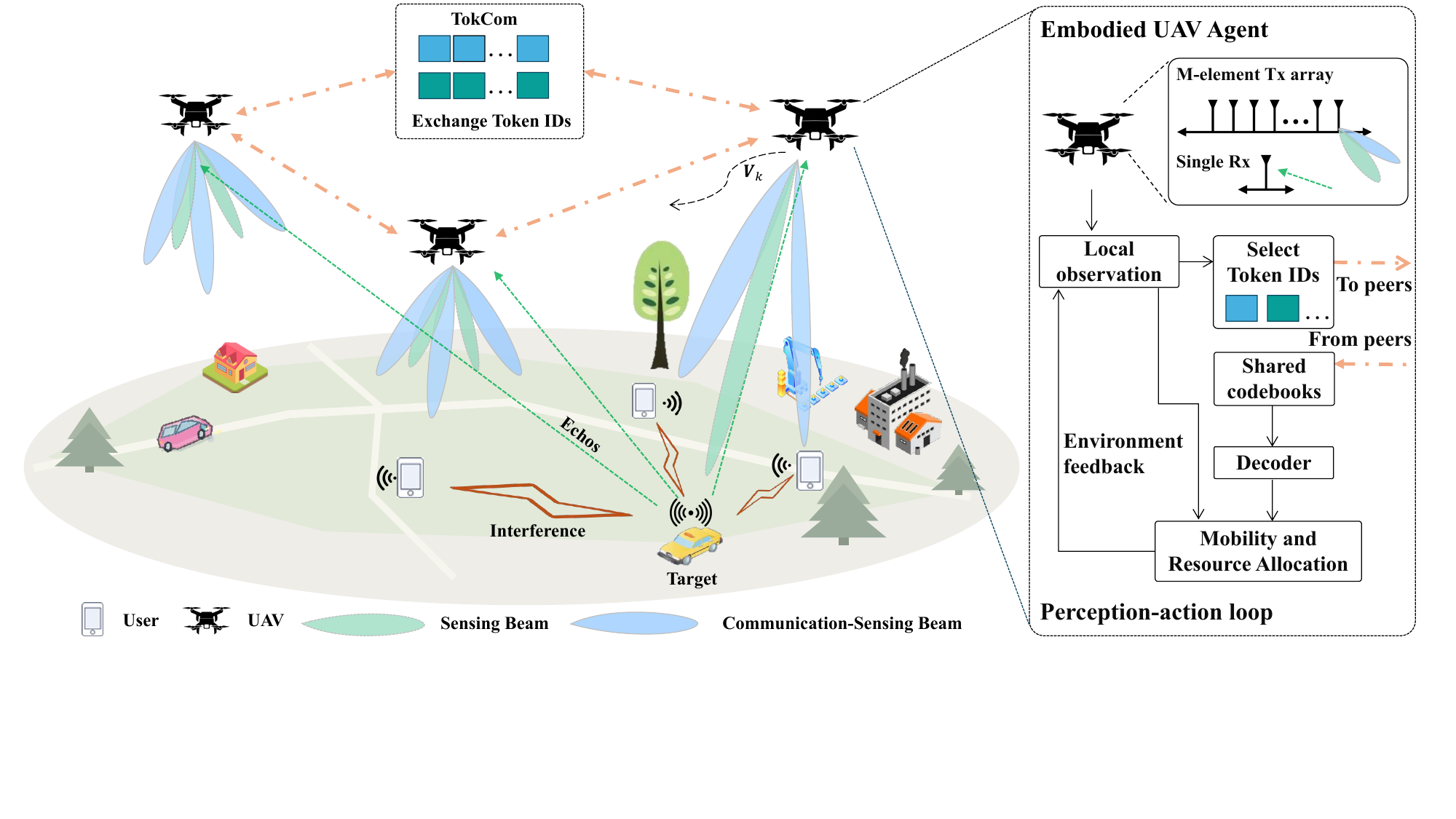}
    \caption{Cooperative multi-UAV ISAC system.}
    \label{fig:system_model}
\end{figure}

The slot-level cooperative protocol is illustrated in Fig.~\ref{fig:protocol}. At the beginning of each time slot, the UAVs exchange compact semantic tokens over air-to-air links.  Based on their local observations and received tokens, the UAVs determine their trajectory and resource allocation decisions in a distributed manner, and subsequently provide downlink communication services to the CUs while cooperatively sensing the STs.
\begin{figure}
    \centering
    \includegraphics[
        width=\linewidth,
        trim={0mm 0cm 0mm 2.5cm},
        clip
    ]{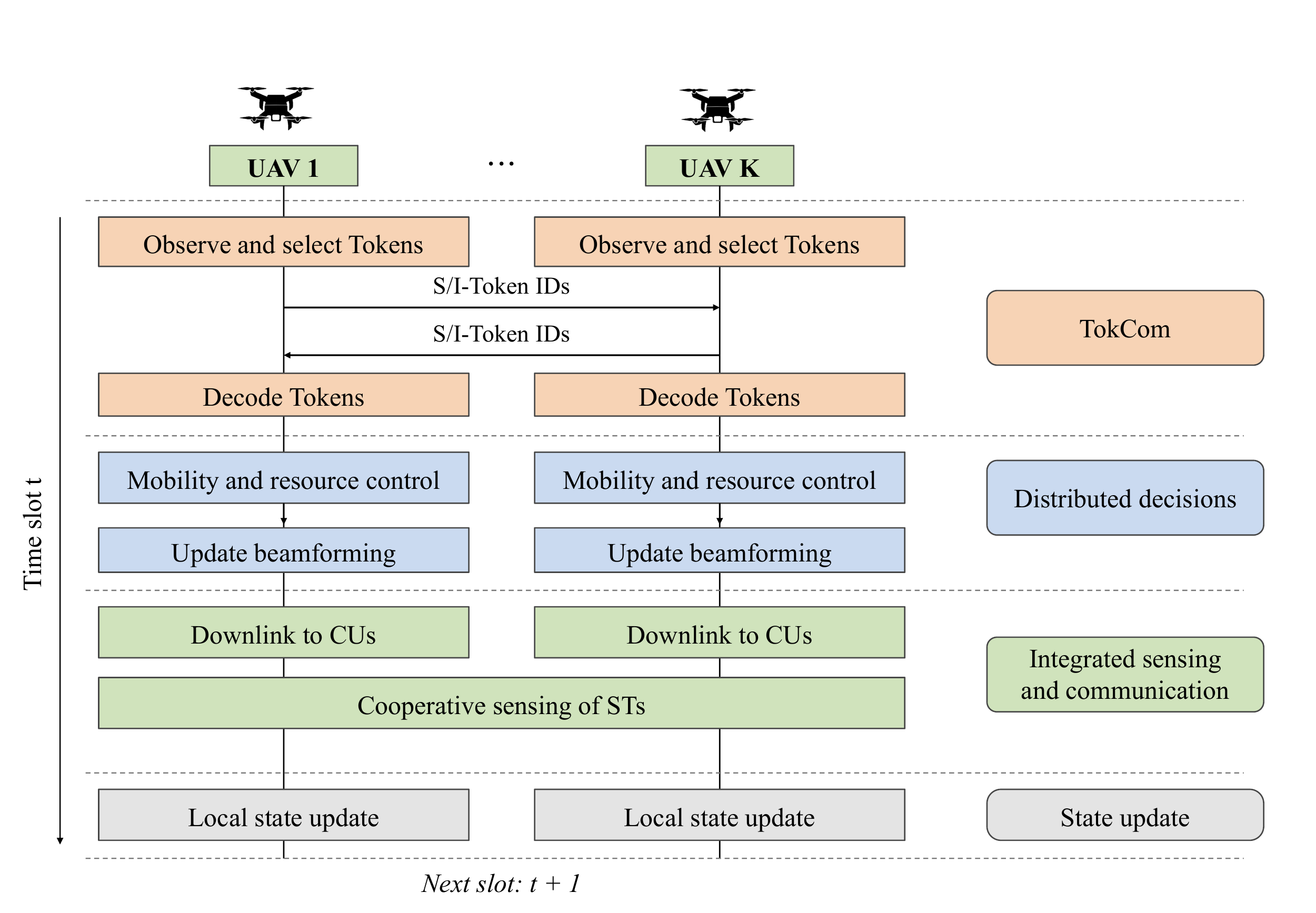}
    \caption{Slot-level cooperative protocol for the multi-UAV ISAC system.}
    \label{fig:protocol}
\end{figure}

\subsection{Channel Model}

Given the dominant line-of-sight (LoS) propagation in UAV-assisted networks, we adopt a far-field LoS channel model. The distances from UAV $k$ to CU $i$ and ST $j$ in time slot $t$ are denoted by $d_{k,i}^{c}[t]=\lVert\mathbf q_k[t]-\mathbf p_i^{c}\rVert_2$ and $d_{k,j}^{s}[t]=\lVert\mathbf q_k[t]-\mathbf p_j^{s}\rVert_2$, respectively. For the $M$-element uniform linear array, the steering vector toward angle $\theta$ is given by~\cite{guo2025flexible}
\begin{equation}
\mathbf a(\theta)=\left[1,e^{\jmath\frac{2\pi d_a}{\lambda}\cos\theta},\ldots,e^{\jmath\frac{2\pi d_a}{\lambda}(M-1)\cos\theta}\right]^{T},
\end{equation}
where $d_a$ and $\lambda$ denote the antenna spacing and carrier wavelength, respectively.
Accordingly, the channel from UAV $k$ to CU $i$ is modeled as
\begin{equation}
\mathbf h_{k,i}^{c}[t]
=\sqrt{\beta_0\left(d_{k,i}^{c}[t]\right)^{-\alpha_{\rm c}}}
\mathbf{a}_{k,i}^{c}  \left(\theta_{k,i}^{c}[t]\right),
\end{equation}
where $\beta_0$ is the reference channel power gain and $\alpha_{\rm c}$ is the path-loss exponent. For monostatic sensing, the round-trip channel between UAV $k$ and ST $j$ is expressed as
\begin{equation}
\mathbf g_{k,j}^{s}[t]
=\rho_{k,j}^{s}[t]  \left( \mathbf{a}_{k,j}^{s}\left(\theta_{k,j}^{s}[t]\right) \right)^{H}  ,
\end{equation}
where $\rho_{k,j}^{s}[t] = \sqrt{\beta_{\rm s}\left(d_{k,j}^{s}[t]\right)^{-4}}$
denotes the monostatic sensing channel amplitude coefficient, and
$\beta_{\rm s}$ is the effective sensing reflection gain.

\subsection{Communication and Sensing Signal Model}
To characterize the association between UAVs and CUs, we define the binary variable
$\alpha_{k,i}[t]\in\{0,1\}$, where $\alpha_{k,i}[t]=1$ indicates that CU $i$
is served by UAV $k$ at time slot $t$, and $\alpha_{k,i}[t]=0$ otherwise.
Each CU is associated with one UAV, satisfying $\sum_{k\in\mathcal{K}}\alpha_{k,i}[t]=1, \forall i\in\mathcal{C}$.
We adopt frequency-division multiple access for downlink transmission, where UAV $k$ allocates an orthogonal
subchannel with bandwidth $b_{k,i}[t]$ to each associated CU $i$. On the
subchannel assigned to CU $i$, UAV $k$ transmits one communication stream and
$N_s$ dedicated sensing streams. The corresponding stream vector is defined as $\mathbf{s}_{k,i}[t] = \left[s_{k,i}[t], s_{k,i,s_1}[t], \ldots, s_{k,i,s_{N_s}}[t] \right]^{T} \in \mathbb{C}^{(N_s+1)\times1}$,
where $s_{k,i}[t]$ denotes the information symbol intended for CU $i$, while
$s_{k,i,s_j}[t]$ denotes the dedicated sensing symbol transmitted by UAV $k$
over subchannel $i$ for ST $j$.
All streams are assumed to be mutually independent complex Gaussian random
variables with zero mean and unit variance, satisfying $\mathbb{E}\{ \mathbf{s}_{k,i}[t] {\mathbf{s}_{k,i}[t]}^{H} \} = \mathbf{I}$. The corresponding transmit beamforming matrix is expressed as
\begin{equation}
    \begin{aligned}
    \mathbf{W}_{k,i}[t]
    =
    &[
    \underbrace{
    \mathbf{w}_{k,i}[t]
    }_{\substack{\text{for communication and sensing}}}, \\
    &\underbrace{
    \mathbf{w}_{k,i,s_1}[t],
    \ldots,
    \mathbf{w}_{k,i,s_{N_s}}[t]
    }_{\text{dedicated for sensing}}
    ]
    \in\mathbb{C}^{M\times(N_s+1)}.    
    \end{aligned}
\end{equation}
Thus, the received signal at CU $i$ served by UAV $k$ is
\begin{equation}
\begin{aligned}
y_{k,i}^{c}[t]
&=
{\mathbf{h}_{k,i}^{c}[t]}^{H}
\mathbf{w}_{k,i}[t]s_{k,i}[t] \\
&+
\sum_{r=1}^{N_s}
{\mathbf{h}_{k,i}^{c}[t]}^{H}
\mathbf{w}_{k,i,s_r}[t]s_{k,i,s_r}[t]
+
n_{k,i}^{c}[t],
\end{aligned}
\end{equation}
where $n_{k,i}^{c}[t]\sim
\mathcal{CN}(0,N_0b_{k,i}[t])$ denotes the additive white Gaussian noise, and $N_0$ is its power spectral density.
Accordingly, the achievable communication rate over the link between UAV $k$
and CU $i$ is expressed in terms of the signal-to-interference-plus-noise ratio (SINR) as
\begin{equation}
R_{k,i}^{c}[t]
=
b_{k,i}[t]
\log_2\left(
1+\mathrm{SINR}_{k,i}[t]
\right),
\end{equation}
where
\begin{equation}
\mathrm{SINR}_{k,i}[t]
=
\frac{
\left|
{\mathbf{h}_{k,i}^{c}[t]}^{H}
\mathbf{w}_{k,i}[t]
\right|^2
}{
\displaystyle
\sum_{r=1}^{N_s}
\left|
{\mathbf{h}_{k,i}^{c}[t]}^{H}
\mathbf{w}_{k,i,s_r}[t]
\right|^2
+
N_0b_{k,i}[t]
}.
\end{equation}
The communication rate of CU $i$ is $R_i^{c}[t]=\sum_{k=1}^{K}\alpha_{k,i}[t]R_{k,i}^{c}[t]$.
Meanwhile, both the communication and dedicated sensing streams are exploited for monostatic sensing. We assume that the echoes reflected by different STs have resolvable round-trip delays and can therefore be separated through delay-domain matched filtering. Hence, no inter-target interference remains after signal processing. For an active subchannel assigned to CU $i$, the processed echo associated with ST $j$ at UAV $k$ is given by
\begin{equation}
y_{k,j,i}^{s}[t]
=
\mathbf g_{k,j}^{s}[t]
\mathbf{W}_{k,i}[t]\mathbf{s}_{k,i}[t]
+
n_{k,j,i}^{s}[t],
\end{equation}
where $n_{k,j,i}^{s}[t] \sim \mathcal{CN}(0,N_0b_{k,i}[t])$ denotes the sensing noise. Echoes from different STs can be separated in the delay domain owing to their resolvable propagation-delay differences, thereby suppressing inter-target interference~\cite{jing2024isac}. Accordingly, the sensing signal-to-noise ratio (SNR) is given by
\begin{equation}
\gamma_{k,j,i}^{s}[t]
=
\frac{
\alpha_{k,i}[t] G_p
\left\|
\rho_{k,j}^{s}[t]
{\mathbf{a}_{k,j}^{s}[t]}^{H}
\mathbf{W}_{k,i}[t]
\right\|_2^2
}{
N_0b_{k,i}[t]
},
\end{equation}
where $G_p$ is the processing gain of the matched filter.
Accordingly, the aggregate sensing rate contributed by subchannel $i$ of UAV
$k$ is given by~\cite{lyu2025movable}
\begin{equation}
R_{k,i}^{s}[t]
=
\sum_{j=1}^{N_s}
b_{k,i}[t]
\log_2\left(
1+\gamma_{k,j,i}^{s}[t]
\right).
\end{equation}
The contribution for ST $j$ is denoted by $R_{k,j,i}^{s}[t]=b_{k,i}[t]\log_2(1+\gamma_{k,j,i}^{s}[t])$, so that $R_{k,i}^{s}[t]=\sum_{j=1}^{N_s}R_{k,j,i}^{s}[t]$. The overall sensing rate is then expressed as $R^{s}[t] = \sum_{k=1}^{K} \sum_{i=1}^{N_c} \alpha_{k,i}[t]R_{k,i}^{s}[t]$.

\subsection{TokCom for Inter-UAV Coordination}
To alleviate the partial observability of decentralized UAV control under a fixed communication budget, each UAV broadcasts a compact semantic message to the other UAVs. The message consists of two complementary token types: the S-Token, which summarizes the current local state, and the I-Token, which conveys the locally intended behavior. The set of token types is denoted by $\mathcal{G}=\{\mathrm{S},\mathrm{I}\}$.
For each token type $g\in\mathcal{G}$, all UAVs share an offline-pretrained semantic codebook
\begin{equation}
\mathcal{E}_g
=
\left\{
\mathbf{e}_{g,n}\in\mathbb{R}^{D_g^{\rm tok}}
\right\}_{n=1}^{N_g^{\rm tok}},
\end{equation}
where $\mathbf{e}_{g,n}$ denotes the $n$th embedding in $\mathcal{E}_g$, while $N_g^{\rm tok}$ and $D_g^{\rm tok}$ denote the corresponding codebook size and embedding dimension, respectively. During online policy learning, the codebooks are frozen, and UAV $k$ selects the token index in each slot according to
\begin{equation}
m_{k,g,l}[t]
\sim
\pi_{\theta_{k,g}}^{\rm tok}
\left(
\cdot\mid\mathbf{o}_k[t]
\right),
\quad
m_{k,g,l}[t]\in\{1,\ldots,N_g^{\rm tok}\},
\end{equation}
where $\pi_{\theta_{k,g}}^{\rm tok}$ denotes the token selection and composition policy of UAV $k$ for type $g$, parameterized by $\theta_{k,g}$, and $\mathbf{o}_k[t]$ is the local observation of UAV $k$ at time slot $t$. The variable $m_{k,g,l}[t]$ denotes the token index selected for the $l$th message slot of type $g$. These selected indices form the type-$g$ token sequence generated by UAV $k$, given by
\begin{equation}
\mathbf{m}_{k,g}[t] = \left[ m_{k,g,1}[t], \ldots, m_{k,g,L_g}[t] \right],
g\in\mathcal{G},
\end{equation}
where $L_g$ denotes the number of token slots for type $g$. The complete token message is formed by concatenating the S-Token and I-Token sequences as $\mathbf{m}_k[t] = \left[ \mathbf{m}_{k,\mathrm{S}}[t], \mathbf{m}_{k,\mathrm{I}}[t] \right]$. The slot count and the semantic role of each slot remain fixed throughout online training and execution, while different token combinations across these slots convey different message semantics.
The semantic codebooks and decoders are pretrained offline and subsequently frozen to preserve a shared semantic basis. Token selection and composition are jointly optimized with the control policy, allowing context-dependent inter-UAV coordination to emerge through collaborative policy learning.

\subsection{Problem Formulation}


UAV energy consumption mainly consists of the energy required for inter-UAV token transmission and mobility. Owing to the small token payload, the corresponding transmission energy is typically several orders of magnitude lower than the propulsion energy. Therefore, the token transmission energy is neglected, and only the UAV propulsion energy is considered. The propulsion power is modeled as a function of the UAV flight
speed, given by
\begin{equation}
P_k^{\rm fly}[t]
=
P_0
+
\frac{1}{2}
\rho_{\rm air} A C_D
\left\|
\mathbf{v}_k[t]
\right\|_2^3
+
P_i
\sqrt{
1+
\frac{
\left\|
\mathbf{v}_k[t]
\right\|_2^4
}{
4v_0^4
}
},
\end{equation}
where $P_0$ and $P_i$ denote the blade-profile power and induced power,
respectively, $\rho_{\rm air}$ is the air density, $A$ is the rotor-disc area,
$C_D$ is the drag coefficient, and $v_0$ is the mean rotor-induced velocity.
Accordingly, the propulsion energy consumed by UAV $k$ during time slot $t$ is $E_k^{\rm fly}[t] = P_k^{\rm fly}[t]T$.

The total propulsion energy consumed by all UAVs over the considered time
horizon is therefore given by $E^{\rm fly} = \sum_{t=1}^{N_T} \sum_{k=1}^{K} E_k^{\rm fly}[t]$.
We jointly optimize the UAV velocities
$\mathbf{V}=\{\mathbf{v}_k[t]\}$, UAV--CU associations $\mathbf{A}=\{\alpha_{k,i}[t]\}$, bandwidth allocation $\mathbf{B}=\{b_{k,i}[t]\}$, link power allocation $\mathbf{P}=\{p_{k,i}[t]\}$, and transmit beamforming matrices $\mathbf{W}=\{\mathbf{W}_{k,i}[t]\}$. Let $\mathbf{M} \triangleq \{ m_{k,g,l}[t] \mid k\in\mathcal{K}, g\in\mathcal{G}, l=1,\ldots,L_g\}$
denote the collection of token indices selected by all UAVs, and let $\mathcal{E}\triangleq\{\mathcal{E}_g\}_{g\in\mathcal{G}}$ denote the collection of codebooks.
The energy-minimization problem is formulated as
\begin{subequations}
\label{problem:P1}
\begin{align}
\mathcal{P}_1 \quad
&\underset{\mathcal{E},\mathbf{M},\mathbf{V},\mathbf{A},
\mathbf{B},\mathbf{P},\mathbf{W}}{\operatorname{min}}
\quad \sum_{t=1}^{N_T}\sum_{k=1}^{K}E_k^{\rm fly}[t]
\label{problem:P1_obj}
\\
\operatorname{s.t.}\quad
& R_i^{c}[t] \geq R_{\min,i}^{c}, \forall i,t,
\label{problem:P1_c1}
\\
& R^{s}[t] \geq R_{\min}^{s}, \forall t,
\label{problem:P1_c2}
\\
& \sum_{k=1}^{K}\alpha_{k,i}[t]=1,
\alpha_{k,i}[t]\in\{0,1\}, \forall k,i,t,
\label{problem:P1_c3}
\\
& \sum_{i=1}^{N_c}b_{k,i}[t]\leq B_k, \forall k,t,
\label{problem:P1_c4}
\\
& \sum_{i=1}^{N_c}p_{k,i}[t]\leq P_k^{\max}, \forall k,t,
\label{problem:P1_c6}
\\
& \left\|\mathbf{W}_{k,i}[t]\right\|_F^2
\leq p_{k,i}[t], \forall k,i,t,
\label{problem:P1_c7}
\\
& \left\|\mathbf{v}_k[t]\right\|_2
\leq V_k^{\max}, \forall k,t,
\label{problem:P1_c9}
\\
& \left\|\mathbf{q}_k[t]-\mathbf{q}_{\ell}[t]\right\|_2
\geq D_{\min}, \forall k\neq\ell,t,
\label{problem:P1_c11}
\end{align}
\end{subequations}
where $R_{\min,i}^{c}$ is the minimum required communication rate of CU $i$, and $R_{\min}^{s}$ is the minimum required aggregate sensing rate. $B_k$, $P_k^{\max}$, and $V_k^{\max}$ specify the bandwidth, transmit power, and speed limits of UAV $k$, respectively, while $D_{\min}$ denotes the minimum inter-UAV separation. Constraints \eqref{problem:P1_c1} and \eqref{problem:P1_c2} impose the communication and sensing requirements, respectively. Constraint \eqref{problem:P1_c3} associates each CU with exactly one UAV in each time slot. Constraints \eqref{problem:P1_c4} and \eqref{problem:P1_c6} enforce the bandwidth and power budgets of each UAV, while \eqref{problem:P1_c7} bounds each link's beamforming power by its allocated power. Finally, \eqref{problem:P1_c9} limits UAV speed, and \eqref{problem:P1_c11} maintains the prescribed inter-UAV separation.
Problem $\mathcal{P}_1$ involves discrete association and token decisions coupled with UAV mobility and resource allocation. Moreover, decentralized execution requires each UAV to coordinate its decisions using local observations and limited information exchange.

\section{Multi-Agent Embodied Policy-Driven SI-TokCom Framework}

\subsection{Overview and Dec-POMDP Formulation}

To address the coupled optimization problem $\mathcal{P}_1$, we propose an SI-TokCom framework driven by multi-agent embodied policy learning. Treating each UAV as an embodied agent, we formulate the sequential decision process as a Dec-POMDP $\left\langle \mathcal{K},s_t,\{\widetilde{\mathbf{o}}_k\}_{k\in\mathcal{K}},\{\mathbf{u}_k\}_{k\in\mathcal{K}},\{r_k\}_{k\in\mathcal{K}},\gamma \right\rangle$, where $s_t\in\mathcal X$ is the global state at time slot $t$, $\widetilde{\mathbf{o}}_k$, $\mathbf{u}_k$, and $r_k$ denote the state, action, and reward of UAV $k \in \mathcal{K}$, respectively, and $\gamma$ is the discount factor. The detailed formulations are as follows.



\subsubsection{State Space}
At time slot $t$, the state of UAV $k$ can be expressed as
$\widetilde{\mathbf{o}}_k[t] = \left[\mathbf{o}_k[t],\mathbf{c}_k[t]\right]$,
where $\mathbf{c}_k[t]$ is the aggregated semantic representation constructed from the received token messages, given by 
\begin{equation}
\mathbf{c}_k[t] = \mathop{\mathrm{Concat}}_{\ell\in\mathcal{K}\setminus\{k\}} \left( \mathcal{D}_{\mathrm{S}} \left( \mathbf{m}_{\ell,\mathrm{S}}[t] \right), \mathcal{D}_{\mathrm{I}} \left( \mathbf{m}_{\ell,\mathrm{I}}[t] \right) \right),
\end{equation}
where $\mathcal{D}_{\mathrm{S}}(\cdot)$ and $\mathcal{D}_{\mathrm{I}}(\cdot)$ denote the frozen semantic decoders. Moreover, $\mathbf{o}_k[t]$ can be expressed as
\begin{equation}
\mathbf{o}_k[t] = \left[
\mathbf{o}_k^{\rm self}[t],
\left\{\mathbf{o}_{k,i}^{c}[t]\right\}_{i\in\mathcal{C}},
\left\{\mathbf{o}_{k,j}^{s}[t]\right\}_{j\in\mathcal{S}}
\right]
\in\mathcal{O}_k,
\label{eq:local_observation}
\end{equation}
where $\mathbf{o}_k^{\rm self}[t]=[\mathbf{q}_k[t],\mathbf{v}_k[t],E_k^{\rm r}[t],\boldsymbol{\alpha}_k[t-1],\mathbf{b}_k[t-1],
\allowbreak
\mathbf{p}_k[t-1]]$. Here, $E_k^{\rm r}[t]$ denotes the remaining onboard energy of UAV $k$, $\boldsymbol{\alpha}_k[t-1]\triangleq\left[\alpha_{k,i}[t-1]\right]_{i\in\mathcal{C}}$, $\mathbf{b}_k[t-1]\triangleq\left[b_{k,i}[t-1]\right]_{i\in\mathcal{C}}$, and $\mathbf{p}_k[t-1]\triangleq\left[p_{k,i}[t-1]\right]_{i\in\mathcal{C}}$. 
Moreover, $\mathbf{o}_{k,i}^{c}[t]=\mathbf{p}_i^{c}-\mathbf{q}_k[t]$ and $\mathbf{o}_{k,j}^{s}[t]=\mathbf{p}_j^{s}-\mathbf{q}_k[t]$ characterize the relative positions of CU $i$ and ST $j$, respectively, where $\hat{\mathbf{p}}_j^{s}[t]$ denotes the estimated position of ST $j$.



\subsubsection{Action Space}
The action of UAV $k$ is defined as $\mathbf{u}_k[t] = \{
\{ m_{k,g,l}[t] \}_{\forall g,\forall l}, \mathbf{v}_k[t], \{ \alpha_{k,i}[t], b_{k,i}[t], p_{k,i}[t], \mathbf{W}_{k,i}[t] \}_{i\in\mathcal{C}}
\}$. The beamforming matrices are included in the joint action but updated via closed-form solutions, avoiding the difficulty of learning complex-valued outputs.

\subsubsection{Surrogate Reward}

For policy learning, we adopt a penalty-relaxed surrogate of $\mathcal{P}_1$ that penalizes communication and sensing rate shortfalls to encourage constraint satisfaction. The surrogate reward of UAV $k$ is defined as
\begin{equation}
\begin{aligned}
r_k[t]
=&
w_{\rm c}\sum_{i=1}^{N_c}R_i^{c}[t]
+
w_{\rm s}R^{s}[t]
-
w_{\rm E}E^{\rm fly}[t]
-
\mathcal{P}_k[t]
\\
&+
\eta
\left(
c_k[t]
-
\frac{1}{K}\sum_{\ell=1}^{K}c_\ell[t]
\right),
\end{aligned}
\label{eq:reward_function}
\end{equation}
where $E^{\rm fly}[t]\triangleq\sum_{\ell=1}^{K}E_\ell^{\rm fly}[t]$. The penalty term $\mathcal{P}_k[t]$ is given by $\mathcal{P}_k[t] = \frac{w_{\rm c}^{\rm P}}{N_c}
\sum_{i=1}^{N_c} \left[ 1-\frac{R_i^{c}[t]}{R_{\min,i}^{c}}
\right]^+ 
+ w_{\rm s}^{\rm P}
\left[
1-\frac{R^{s}[t]}{R_{\min}^{s}}
\right]^+
+ \mathcal{P}_k^{\rm phy}[t]$,
where $\mathcal{P}_k^{\rm phy}[t]$ penalizes the mobility, safety-distance, bandwidth, power, and energy constraint violations associated with UAV $k$. The term $c_k[t]$ measures its local communication and sensing contributions, flight energy, and service load. The local term improves credit assignment without changing the average team reward.

Building on the above Dec-POMDP formulation, the proposed SI-TokCom framework comprises the following two stages, as shown in Fig.~\ref{fig.algrithom_framework}:

\begin{itemize}
\item \emph{Offline codebook pretraining:} The S-Token and I-Token codebooks, together with their semantic encoders and decoders, are pretrained using an expert dataset and frozen to preserve consistent semantic mappings.

\item \emph{Online joint optimization:}
To avoid directly learning complex-valued beamforming matrices,
$\mathcal{P}_1$ is decomposed into two variable blocks.
For fixed $\mathbf{W}$, the token-interactive policies optimize
token selection and determine
the physical actions $\mathbf{V}$, $\mathbf{A}$, $\mathbf{B}$,
and $\mathbf{P}$.
With these physical actions fixed, $\mathbf{W}$ is optimized via FP and KKT-based closed-form updates. The updated $\mathbf{W}$ is then fed back to the policy learning block for the next iteration.
\end{itemize}

\begin{figure*}[t]
    \centering
    \includegraphics[width=\textwidth]{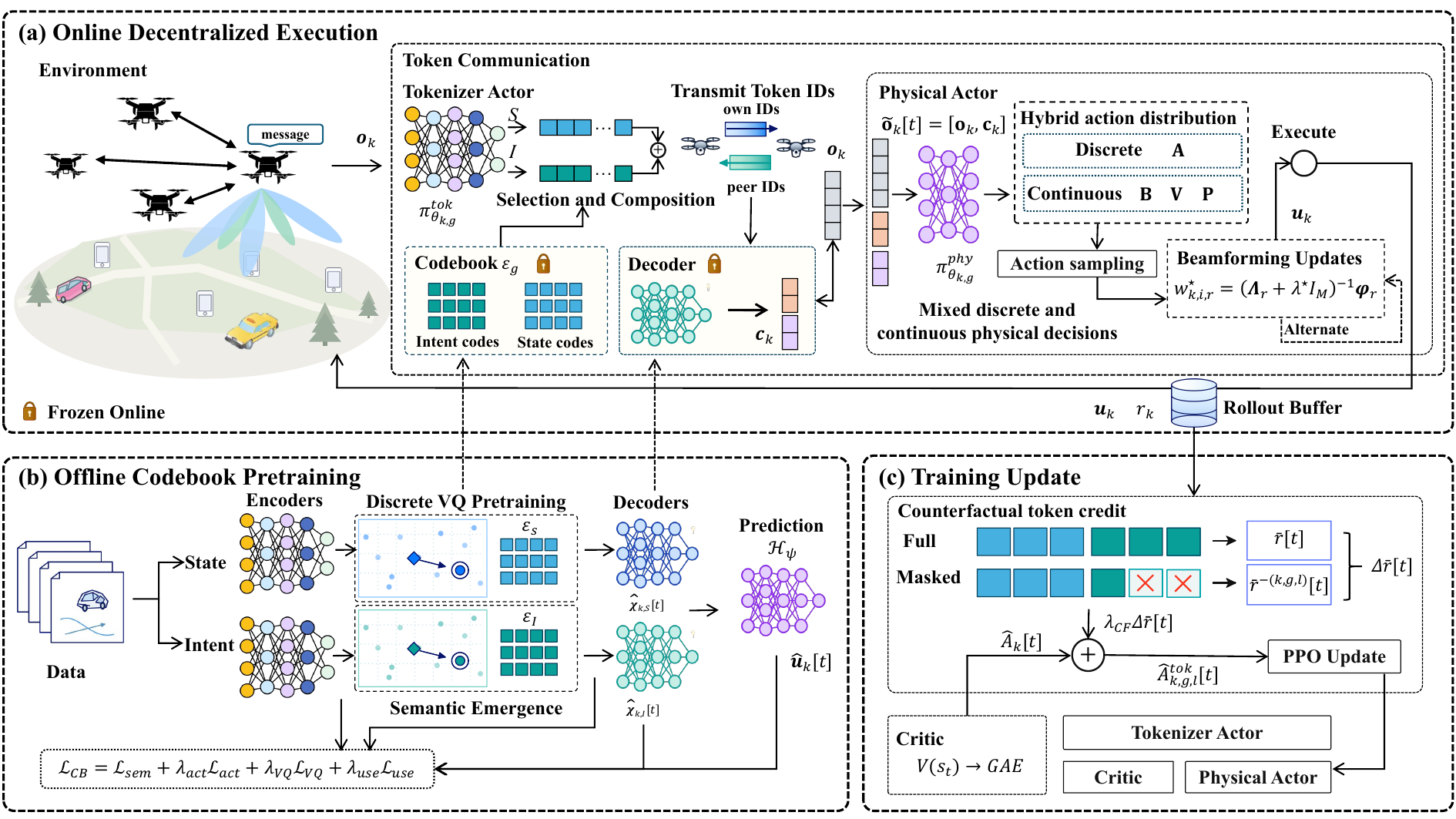}
    \caption{Overview of the proposed framework with offline codebook pretraining and online joint optimization.}
    \label{fig.algrithom_framework}
\end{figure*}

\subsection{Offline S-Token and I-Token Codebook Pretraining}
\label{subsec:offline_codebook}


As illustrated in Fig.~\ref{fig.algrithom_framework}, the S-Token and I-Token codebooks are pretrained using the expert dataset $\mathcal{D}_{\rm E}$. The dataset contains paired samples obtained by evaluating pretrained centralized and local policies from the same environment states. Half of the rollout episodes follow the centralized policy, and the other half follow the local policy. 
 
For each token type $g\in\mathcal{G}$, the encoder produces a continuous representation $\mathbf{z}_{k,g,l}[t]$. VQ maps this representation to its nearest codeword, whose index is given by

\begin{equation}
m_{k,g,l}[t]
=
\underset{n\in\{1,\ldots,N_g^{\rm tok}\}}
{\operatorname{arg\,min}}
\left\|
\mathbf{z}_{k,g,l}[t]
-
\mathbf{e}_{g,n}
\right\|_2^2.
\label{eq:vq_assignment}
\end{equation}
The quantized representation is $\overline{\mathbf{z}}_{k,g,l}[t]=\mathbf{e}_{g,m_{k,g,l}[t]}$. The VQ loss is defined as

\begin{equation}
\begin{aligned}
\mathcal{L}_{\rm VQ}
=
\mathbb{E}_{\mathcal{D}_{\rm E}}
\Bigg[
\sum_{g\in\mathcal{G}}
\sum_{l=1}^{L_g}
&
\left\|
\operatorname{sg}
\left[
\mathbf{z}_{k,g,l}[t]
\right]
-
\overline{\mathbf{z}}_{k,g,l}[t]
\right\|_2^2\\
+
&
\beta_g
\left\|
\mathbf{z}_{k,g,l}[t]
-
\operatorname{sg}
\left[
\overline{\mathbf{z}}_{k,g,l}[t]
\right]
\right\|_2^2
\Bigg],
\end{aligned}
\label{eq:vq_loss}
\end{equation}
where $\beta_g$ is the commitment coefficient.

Let $\widehat{\boldsymbol{\chi}}_{k,g}[t]=\mathcal{D}_g(\mathbf{m}_{k,g}[t])$ denote the decoded semantic representation and $\boldsymbol{\chi}_{k,g}[t]$ denote the corresponding task-relevant target obtained from $\mathcal{D}_{\rm E}$. The semantic reconstruction loss is given by
\begin{equation}
\mathcal{L}_{\rm sem}
=
\mathbb{E}_{\mathcal{D}_{\rm E}}
\left[
\sum_{g\in\mathcal{G}}
\left\|
\widehat{\boldsymbol{\chi}}_{k,g}[t]
-
\boldsymbol{\chi}_{k,g}[t]
\right\|_2^2
\right].
\label{eq:semantic_loss}
\end{equation}

To preserve decision-relevant information, an auxiliary predictor
$\mathcal{H}_{\psi}$ with parameters $\psi$ reconstructs the expert
action from the decoded state and intent semantics.
The action-consistency loss is defined as
\begin{equation}
\mathcal{L}_{\rm act}
=
\mathbb{E}_{\mathcal{D}_{\rm E}}
\left[
\ell_{\rm act}
\left(
\mathcal{H}_{\psi}
\left(
\left[
\widehat{\boldsymbol{\chi}}_{k,\mathrm{S}}^{T}[t],
\widehat{\boldsymbol{\chi}}_{k,\mathrm{I}}^{T}[t]
\right]^{T}
\right),
\mathbf{a}_k^{\star}[t]
\right)
\right],
\label{eq:action_loss}
\end{equation}
where $\mathbf{a}_k^{\star}[t]$ denotes the expert physical action
provided in $\mathcal{D}_{\rm E}$, and $\ell_{\rm act}(\cdot)$
combines classification and regression losses for discrete
and continuous control decisions, respectively.

To prevent codebook collapse, the codeword-usage distribution is regularized as
\begin{equation}
\mathcal{L}_{\rm use}
=
\sum_{g\in\mathcal{G}}
D_{\rm KL}
\left(
\overline{\mathbf{p}}_g
\middle\|
\mathbf{u}_g
\right),
\label{eq:codebook_usage_loss}
\end{equation}
where $\overline{\mathbf{p}}_g$ is the empirical codeword-usage distribution and $\mathbf{u}_g$ is the corresponding uniform distribution. 
The overall offline training loss is
\begin{equation}
\mathcal{L}_{\rm CB}
=
\mathcal{L}_{\rm sem}
+
\lambda_{\rm act}\mathcal{L}_{\rm act}
+
\lambda_{\rm VQ}\mathcal{L}_{\rm VQ}
+
\lambda_{\rm use}\mathcal{L}_{\rm use},
\label{eq:codebook_training_loss}
\end{equation}
where $\lambda_{\rm act}$, $\lambda_{\rm VQ}$, and $\lambda_{\rm use}$ are weighting coefficients. This joint training encourages task-relevant semantics to emerge in the codebooks. After pretraining, the codebooks and semantic decoders are frozen to preserve the token-index-to-semantics mapping during online policy optimization.

\subsection{Online Joint Optimization}
\label{subsec:online_optimization}

With the pretrained codebooks and semantic decoders fixed, the online stage alternates between two optimization blocks. The beamforming block updates $\mathbf W$ for given $\mathbf V$, $\mathbf A$, $\mathbf B$, and $\mathbf P$ using closed-form updates. The policy block learns token selection and composition while jointly optimizing the mixed discrete and continuous physical decisions through MAPPO.

\subsubsection{Beamforming Optimization Based on Closed-Form Updates}
\label{subsec:beamforming_optimization}

Since the propulsion energy is independent of $\mathbf W$, the beamforming subproblem maximizes a dynamically weighted sensing and communication utility:
\begin{subequations}
\label{problem:SP1}
\begin{align}
\mathcal P_2:\quad
\max_{\mathbf W_{k,i}[t]}\quad &
\omega_1\sum_{j=1}^{N_s}R_{k,j,i}^{s}[t]
+\omega_2R_{k,i}^{c}[t]
\label{problem:SP1_obj}\\
\operatorname{s.t.}\quad &
\left\|\mathbf W_{k,i}[t]\right\|_F^2
\leq p_{k,i}[t].
\label{problem:SP1_c1}
\end{align}
\end{subequations}
Here, the base weights $\omega_1,\omega_2\geq0$ satisfy $\omega_1+\omega_2=1$ and control the relative priorities of sensing and communication, respectively. To address the non-convexity of $\mathcal P_2$, we apply the Lagrangian dual transform and quadratic transform under the FP framework. The resulting beamforming update is characterized below.

\begin{theorem}
\label{thm:beamforming_update}
For fixed weights and FP auxiliary variables, and $p_{k,i}[t]>0$, an optimal solution to the resulting quadratic subproblem is
\begin{equation}
\mathbf w_{k,i,r}^{\star}
=
\left(
\boldsymbol{\Lambda}_r+\lambda^{\star}\mathbf I_M
\right)^{-1}
\boldsymbol{\varphi}_r,
\quad r=1,\ldots,N_s+1,
\label{eq:beamforming_closed_form}
\end{equation}
where $\boldsymbol{\varphi}_r$ and $\boldsymbol{\Lambda}_r$ are defined in \eqref{eq:FP_phi} and \eqref{eq:FP_Lambda}, respectively, and $\lambda^{\star}\geq0$ is the optimal Lagrange multiplier associated with the power constraint. 
\end{theorem}
\noindent\emph{Proof:} See Appendix~\ref{app:fp_beamforming} for more details.
\hfill$\square$\par

\begin{figure*}[!t]
\begin{equation}  
\boldsymbol{\varphi}_r[t] = \begin{cases}  
\sum_{j=1}^{N_s} \omega_1 \alpha_{k,i}[t] b_{k,i}[t]  
\sqrt{1+\mu_j^{s}} \sqrt{G_p}  
\left(\rho_{k,j}^{s}\right)^*  
\left(\xi_{j,r}^{s}\right)^*  
\mathbf{a}_{k,j}^{s}  
+  
\omega_2 b_{k,i}[t]  
\sqrt{1+\mu^{c}}  
\left(\xi^{c}\right)^*  
\mathbf{h}_{k,i}^{c},  
 r=1, \\  
\sum_{j=1}^{N_s}  
\omega_1 \alpha_{k,i}[t] b_{k,i}[t]  
\sqrt{1+\mu_j^{s}}  
\sqrt{G_p}  
\left(\rho_{k,j}^{s}\right)^*  
\left(\xi_{j,r}^{s}\right)^*  
\mathbf{a}_{k,j}^{s},  
 r=2,\ldots,N_s+1  
\end{cases}.
\label{eq:FP_phi}  
\end{equation}

\begin{equation} 
\begin{aligned} 
\boldsymbol{\Lambda}_r[t] =  
& \sum_{j=1}^{N_s} \omega_1 b_{k,i}[t] \left\| \boldsymbol{\xi}_j^{s} \right\|_2^2 
\alpha_{k,i}[t]G_p \left| \rho_{k,j}^{s} \right|^2 \mathbf{a}_{k,j}^{s} 
\left( \mathbf{a}_{k,j}^{s} \right)^{H} + \omega_2 b_{k,i}[t] \left| \xi^{c} \right|^2 \mathbf{h}_{k,i}^{c} \left( \mathbf{h}_{k,i}^{c} \right)^{H}. 
\end{aligned} 
\label{eq:FP_Lambda} 
\end{equation}
\hrulefill
\end{figure*}

To account for communication and sensing quality-of-service (QoS) violations, the effective weights at outer iteration $\ell$ are updated as
\begin{equation}
\omega_1^{{\rm eff},(\ell)}
=
\omega_1+\eta_{\rm s}^{{\rm BF},(\ell)},
\qquad
\omega_2^{{\rm eff},(\ell)}
=
\omega_2+\eta_{\rm c}^{{\rm BF},(\ell)},
\label{omega_1_2update}
\end{equation}
where
$\eta_{\rm s}^{{\rm BF},(\ell+1)}
=
\min\{
\eta_{{\rm s},\max}^{\rm BF},
\eta_{\rm s}^{{\rm BF},(\ell)}
+\tau_{\rm s}
[1-R^{s,(\ell)}/R_{\min}^{s}]^+
\}$
and
$\eta_{\rm c}^{{\rm BF},(\ell+1)}
=
\min\{
\eta_{{\rm c},\max}^{\rm BF},
\eta_{\rm c}^{{\rm BF},(\ell)}
+\tau_{\rm c}
[1-\min_{i\in\mathcal C}\{R_i^{c,(\ell)}/R_{\min,i}^{c}\}]^+
\}$.
An unsatisfied QoS requirement therefore receives greater priority in the subsequent beamforming update. The outer iteration terminates when both requirements are satisfied or the maximum number of iterations is reached. The optimization based on closed-form updates is summarized in Algorithm~\ref{alg:FP_beamforming}.

\begin{algorithm}[t]
\caption{Beamforming Optimization Based on Closed-Form Updates}
\label{alg:FP_beamforming}
\begin{algorithmic}[1]
\STATE \textbf{Input:} Fixed $\mathbf V$, $\mathbf A$, $\mathbf B$, and $\mathbf P$; channel information
\STATE Initialize feasible $\mathbf W$ and $\eta_{\rm s}^{\rm BF}=\eta_{\rm c}^{\rm BF}=0$
\FOR{$\ell=0$ to $L_{\rm out}-1$}
    \STATE Update $\omega_1^{{\rm eff},(\ell)}$ and $\omega_2^{{\rm eff},(\ell)}$ using \eqref{omega_1_2update}
    \FOR{$n=0$ to $L_{\rm in}-1$}
        \STATE Update the FP auxiliary variables $\boldsymbol{\mu}$ and $\boldsymbol{\xi}$
        \STATE Calculate $\boldsymbol{\varphi}_r$ and $\boldsymbol{\Lambda}_r$
        \STATE Set $\lambda^{\star}=0$ if $h(0)\leq0$; otherwise, obtain $\lambda^{\star}$ by bisection
        \STATE Update $\mathbf w_{k,i,r}^{\star}$ using \eqref{eq:beamforming_closed_form}
        \IF{the inner FP-KKT iteration converges}
            \STATE \textbf{break}
        \ENDIF
    \ENDFOR
    \STATE Calculate $R^{s}$ and $\{R_i^{c}\}_{i\in\mathcal C}$
    \IF{$R^{s}\geq R_{\min}^{s}$ and $R_i^{c}\geq R_{\min,i}^{c}$ for all $i\in\mathcal C$}
        \STATE \textbf{break}
    \ENDIF
    \STATE Update $\eta_{\rm s}^{\rm BF}$ and $\eta_{\rm c}^{\rm BF}$ from the sensing and communication QoS violations, respectively.
\ENDFOR
\STATE \textbf{Output:} Optimized beamforming matrices $\mathbf W^{\star}$
\end{algorithmic}
\end{algorithm}

\subsubsection{Token Selection and Policy Optimization}
\label{subsec:token_policy_optimization}

The token-selection policy is optimized using proximal policy optimization (PPO).
For the $l$th token slot of type $g$, the token-policy
probability ratio is defined as
\begin{equation}
\varrho_{k,g,l}^{\rm tok}[t]
=
\frac{
\pi_{\theta_{k,g}}^{\rm tok}
\left(m_{k,g,l}[t]\mid\mathbf{o}_k[t]\right)
}{
\pi_{\theta_{k,g,\rm old}}^{\rm tok}
\left(m_{k,g,l}[t]\mid\mathbf{o}_k[t]\right)
},
\label{eq:token_probability_ratio}
\end{equation}
where $\theta_{k,g}$ denotes the current token-policy
parameters of UAV $k$ for type $g$, and $\theta_{k,g,\rm old}$ denotes
the parameters used to collect the rollout data.

To assign credit to individual tokens, we augment the estimated
UAV advantage with a counterfactual reward difference:
\begin{equation}
\widehat{A}_{k,g,l}^{\rm tok}[t]
=
\widehat{A}_k[t]
+
\lambda_{\rm CF}
\left(
\overline{r}[t]
-
\overline{r}^{(-k,g,l)}[t]
\right),
\label{eq:token_advantage}
\end{equation}
where $\widehat{A}_k[t]$ is the estimated advantage of UAV $k$,
$\overline{r}[t]\triangleq K^{-1}\sum_{\ell=1}^{K}r_\ell[t]$
is the average team reward, and $\lambda_{\rm CF}$ weights
the counterfactual contribution.
The reward $\overline{r}^{(-k,g,l)}[t]$ is evaluated by masking
the selected token in the received messages, keeping the
environment state and all other tokens unchanged, and
recomputing the downstream physical actions and beamforming.

The token-selection policy is trained by minimizing
the per-slot clipped surrogate loss
\begin{equation}
\begin{aligned}
&\mathcal L_{\rm tok}(\boldsymbol{\theta})
= -\frac{1}{K} \sum_{k=1}^{K} \mathbb E_t
\Biggl[
\sum_{g\in\mathcal G}
\sum_{l=1}^{L_g}
\min\Bigl\{
\varrho_{k,g,l}^{\rm tok}[t]
\widehat A_{k,g,l}^{\rm tok}[t],
\\
&\quad
\operatorname{clip}
\left(
\varrho_{k,g,l}^{\rm tok}[t],
1-\epsilon,
1+\epsilon
\right)
\widehat A_{k,g,l}^{\rm tok}[t]
\Bigr\}
\Biggr]
-
\lambda_{\rm H}\mathcal H_{\rm tok},
\end{aligned}
\label{eq:token_ppo_objective}
\end{equation}
where $\boldsymbol{\theta}$ collects the token-policy parameters $\{\theta_{k,g}\}_{k\in\mathcal{K},g\in\mathcal G}$ and the physical-policy parameters, $\mathbb{E}_t$ denotes averaging over rollout samples,
$\epsilon$ is the clipping parameter, $\mathcal{H}_{\rm tok}$ is the token-policy entropy, and $\lambda_{\rm H}$ is its regularization coefficient.
The physical policy is trained by minimizing the standard clipped MAPPO loss~\cite{yu2022surprising}.

The centralized critic
$V_{\boldsymbol{\phi},k}(s_t)$,
parameterized by $\boldsymbol{\phi}$,
estimates the expected discounted return of UAV $k$
from the global state $s_t\in\mathcal{X}$.
It is trained by minimizing
\begin{equation}
\mathcal{L}_{\rm critic}(\boldsymbol{\phi})
=
\frac{1}{2K}
\sum_{k=1}^{K}
\mathbb{E}_t
\left[
\left(
V_{\boldsymbol{\phi},k}(s_t)
-
\widehat{V}_k[t]
\right)^2
\right],
\label{eq:critic_loss}
\end{equation}
where $\widehat{V}_k[t]$ denotes the return target
computed from the rollout data.
The centralized critic is used only during training.
The advantage estimates and return targets are computed using generalized advantage estimation (GAE).
The overall online training procedure is summarized in
Algorithm~\ref{alg:online_sitcom}.

%

\begin{algorithm}[t]
\caption{Online Training of the SI-TokCom Framework}
\label{alg:online_sitcom}
\begin{algorithmic}[1]
\STATE \textbf{Input:} Pretrained codebooks and semantic decoders $\{\mathcal E_g,\mathcal D_g\}_{g\in\mathcal G}$
\STATE Freeze $\{\mathcal E_g,\mathcal D_g\}$; initialize policy parameters $\boldsymbol{\theta}$ and critic parameters $\boldsymbol{\phi}$
\FOR{each training episode}
    \STATE Reset the environment and obtain $s_1$ and $\{\mathbf o_k[1]\}$; set $\boldsymbol{\theta}_{\rm old}\leftarrow\boldsymbol{\theta}$
    \FOR{$t=1$ to $N_T$}
        \STATE Each UAV selects $m_{k,g,l}[t]$
        \STATE Exchange $\{\mathbf m_k[t]\}$ and construct $\{\widetilde{\mathbf o}_k[t]\}$ using the frozen decoders
        \STATE Each UAV selects $\mathbf v_k[t]$, $\alpha_{k,i}[t]$, $b_{k,i}[t]$, and $p_{k,i}[t]$ based on $\widetilde{\mathbf o}_k[t]$
        \STATE Update $\mathbf W_{k,i}[t]$ using Algorithm~\ref{alg:FP_beamforming}
        \STATE Execute the joint action $\mathbf u[t]$; observe $\{r_k[t]\}$, $s_{t+1}$, and $\{\mathbf o_k[t+1]\}$
        \STATE Evaluate the masked-token rewards $\overline r^{(-k,g,l)}[t]$ under the same pre-action state
        \STATE Store the rollout data and policy probabilities
    \ENDFOR
    \STATE Compute $\widehat A_k[t]$ and $\widehat V_k[t]$ using GAE
    \STATE Compute $\widehat A_{k,g,l}^{\rm tok}[t]$ using \eqref{eq:token_advantage}
    \FOR{each PPO epoch}
        \STATE Update $\boldsymbol{\theta}$ using \eqref{eq:token_ppo_objective} and the physical-policy MAPPO objective
        \STATE Update $\boldsymbol{\phi}$ by minimizing \eqref{eq:critic_loss}
    \ENDFOR
\ENDFOR
\STATE \textbf{Output:} Trained token-selection and physical-policy parameters $\boldsymbol{\theta}$
\end{algorithmic}
\end{algorithm}

\section{Simulation Results}
\label{sec:simulation}

We evaluate the proposed SI-TokCom framework in a cooperative multi-UAV ISAC system operating in a $600\times600~\mathrm{m}^2$ region with $K=4$ UAVs, $N_c=8$ CUs, and $N_s=3$ STs. The initial horizontal UAV positions are $(150,150)$, $(450,150)$, $(450,450)$, and $(150,450)$~m. The CUs and STs are sampled within $[30,570]\times[30,570]~\mathrm{m}^2$, with each ST initially observable by at least one UAV. The simulations are conducted using the Monte Carlo method. The main parameters are summarized in Table~\ref{tab:simulation_parameters}.
\begin{table}[t]
\centering
\caption{Main Simulation Parameters}
\label{tab:simulation_parameters}
\footnotesize
\renewcommand{\arraystretch}{1.15}
\setlength{\tabcolsep}{2pt}

\begin{tabular*}{\columnwidth}{@{\extracolsep{\fill}}
    >{\raggedright\arraybackslash}m{0.21\columnwidth}
    >{\centering\arraybackslash}m{0.23\columnwidth}
    >{\raggedright\arraybackslash}m{0.30\columnwidth}
    >{\centering\arraybackslash}m{0.18\columnwidth}@{}}
\toprule
\textbf{Parameter} & \textbf{Value}
& \textbf{Parameter} & \textbf{Value} \\
\cmidrule(lr){1-2}\cmidrule(lr){3-4}

$K,\,N_c,\,N_s$ & $4,\,8,\,3$
& $w_{\rm c},\,w_{\rm s},\,w_{\rm E}$ & $0.6,\,0.4,\,0.8$ \\

$H$ & $80~\mathrm{m}$
& $\omega_1,\,\omega_2$ & $0.4,\,0.6$ \\

$N_T,\,T$ & $12,\,1~\mathrm{s}$
& $L_{\rm in},\,L_{\rm out}$ & $20,\,8$ \\

$V_k^{\max}$ & $15~\mathrm{m/s}$
& FP--KKT tolerance & $10^{-3}$ \\

$D_{\min}$ & $20~\mathrm{m}$
& $N_{\mathrm{S}}^{\rm tok},\,N_{\mathrm{I}}^{\rm tok}$
& $64,\,64$ \\

$M$ & $4$
& $D_{\mathrm{S}}^{\rm tok},\,D_{\mathrm{I}}^{\rm tok}$
& $16,\,16$ \\

Carrier frequency & $2.4~\mathrm{GHz}$
& $L_{\mathrm{S}},\,L_{\mathrm{I}}$ & $3,\,13$ \\

$\lambda,\,d_a$ & $0.125~\mathrm{m},\,\lambda/2$
& Intent horizon & $2$ slots \\

$B_k,\,\sum_{k=1}^{K}B_k$ & $(1.25,\,5)~\mathrm{MHz}$
& $\gamma,\,\epsilon$ & $0.95,\,0.2$ \\

$P_k^{\max}$ & $1~\mathrm{W}$
& GAE parameter & $0.95$ \\

$N_0$ & $10^{-17}~\mathrm{W/Hz}$
& Critic learning rate & $5\times10^{-4}$ \\

$\beta_0,\,\alpha_{\rm c}$ & $10^{-3},\,2$
& Mini-batches per update & $2$ \\

$\beta_{\rm s},\,G_p$ & $10^{-5},\,1$
& Offline epochs / batch size & $100\,/\,32$ \\

$R_{\min,i}^{c}$
& $\mathcal{U}[1.84,2.76]$\newline $\mathrm{Mbit/s}$
& Codebook exponential moving average coefficient & $0.99$ \\

$R_{\min}^{s}$ & $20~\mathrm{kbit/s}$
& Training episodes / parallel envs. & $1000\,/\,4$ \\

\bottomrule
\end{tabular*}
\end{table}
We evaluate the following schemes with different information inputs and inter-UAV communication mechanisms.
1)~\textbf{SI-TokCom (TokenCom):} The proposed scheme combines local observations with exchanged S-Tokens and I-Tokens.
2)~\textbf{Local:} Each UAV selects its physical actions solely from its local observation.
3)~\textbf{Centralized:} The policies use global system information, including the states of all UAVs and ground nodes.
4)~\textbf{State Only:} This ablation supplements local observations with S-Tokens only.
5)~\textbf{Intent Only:} This ablation supplements local observations with I-Tokens only.
6)~\textbf{SemanticCom:} This baseline exchanges continuous semantic representations among UAVs to support cooperative control. 

\begin{figure}[t]
    \centering
    \includegraphics[width=1.0\linewidth]{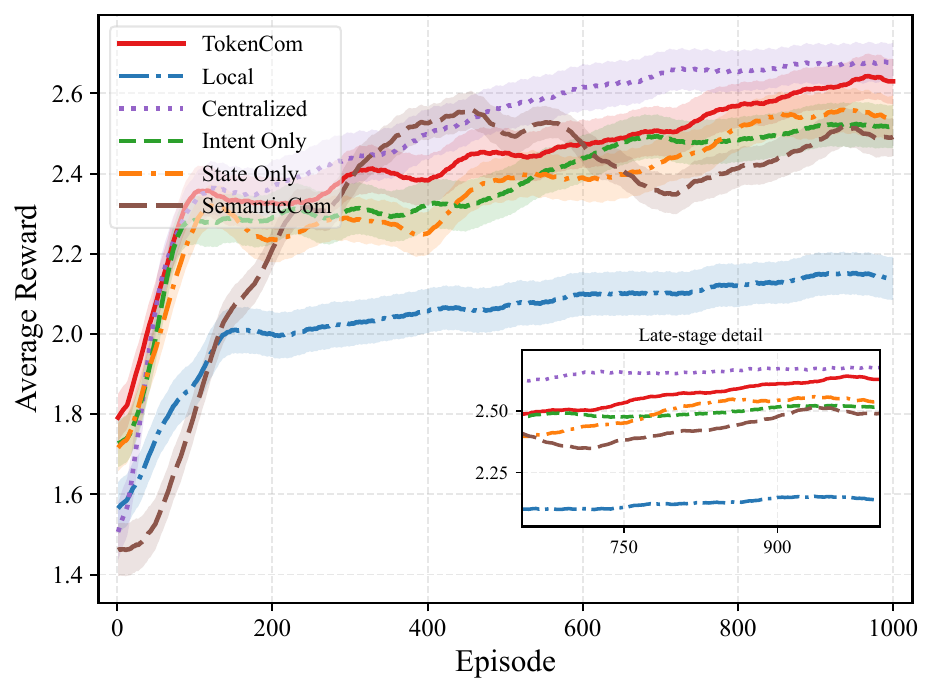}
    \caption{Average reward convergence of different schemes.}
    \label{fig:reward_convergence}
\end{figure}

Fig.~\ref{fig:reward_convergence} compares the average per-slot training reward of different schemes versus the number of episodes. SI-TokCom improves rapidly during the initial stage and attains a higher late-stage reward than Local, SemanticCom, and the two single-token variants, approaching that of Centralized. 
These results suggest that joint state and intent exchange helps alleviate partial observability and improve decentralized coordination. 
At the same message budget, SI-TokCom benefits from decision-oriented state and intent prototypes that better support cooperation than SemanticCom's representations focused on state and action reconstruction. Moreover, compared with continuous messages, discrete tokens enable grouping of information with similar action-value structures during codebook training. Such grouping may help narrow the performance gap between policies under partial and full observability.


\begin{table}[t]
    \centering
    \caption{Late-Stage Performance Comparison}
    \label{tab:performance_comparison}
    \footnotesize
    \renewcommand{\arraystretch}{1.18}
    \setlength{\tabcolsep}{1pt}
    \begin{tabular*}{\columnwidth}{@{\extracolsep{\fill}}lcccccc@{}}
        \toprule
        Scheme
        & \shortstack{Comm.\\rate $\uparrow$\\(Mbit/s)}
        & \shortstack{Sensing\\rate $\uparrow$\\(kbit/s)}
        & \shortstack{Energy $\downarrow$\\(J)}
        & \shortstack{Comm.\\QoS $\uparrow$\\(\%)}
        & \shortstack{Sensing\\QoS $\uparrow$\\(\%)}
        & \shortstack{Coord.\\payload \\(bit)} \\
        \midrule
        Centralized
        & $7.22$ & $65.04$ & $42.18$
        & $99.45$ & $93.48$ & --- \\

        TokenCom
        & $7.14$ & $64.44$ & $42.11$
        & $98.73$ & $93.89$ & $96$ \\

        Local
        & $6.80$ & $44.88$ & $42.11$
        & $99.39$ & $78.91$ & $0$ \\

        Intent Only
        & $7.06$ & $58.60$ & $42.16$
        & $98.83$ & $90.91$ & $78$ \\

        State Only
        & $6.93$ & $64.93$ & $42.23$
        & $97.96$ & $92.46$ & $18$ \\

        SemanticCom
        & $7.03$ & $60.21$ & $43.17$
        & $98.76$ & $90.26$ & $96$ \\
        \bottomrule
    \end{tabular*}

    \vspace{2pt}
    \parbox{\columnwidth}{\footnotesize
        QoS columns report satisfaction rates. Coordination payload
        is measured in bits per UAV per time slot, excluding packet
        headers and channel coding. The coordination overhead of
        Centralized is not modeled.
    }
\end{table}


Table~\ref{tab:performance_comparison} summarizes the average performance of different schemes after convergence. Communication rate is averaged over CUs, while sensing rate and propulsion energy are system totals; these metrics are averaged over time slots within each episode and across environments. With a coordination payload of 96 bits per UAV per time slot, SI-TokCom achieves 98.9\% and 99.1\% of the communication and sensing rates of Centralized, respectively. Compared with Local, it increases the communication and sensing rates by 5.0\% and 43.6\%, respectively, with comparable energy consumption. At the same payload, SI-TokCom achieves a 7.0\% higher sensing rate and 2.5\% lower energy consumption than SemanticCom. Although State Only attains a slightly higher average sensing rate, SI-TokCom achieves the highest sensing QoS satisfaction rate (93.89\%), exceeding State Only and Intent Only. This indicates that a higher average sensing rate does not necessarily imply better satisfaction of sensing requirements.

\begin{figure}[t]
    \centering
    \includegraphics[width=1.0\linewidth]{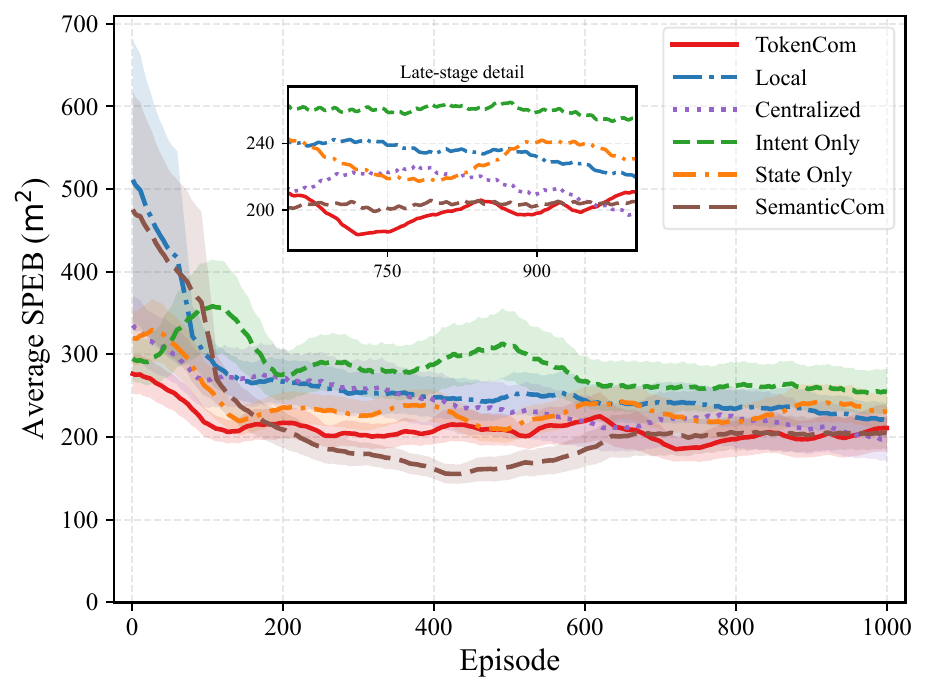}
    \caption{Average SPEB over all STs during training for different schemes.}
    \label{fig:speb_convergence}
\end{figure}


Fig.~\ref{fig:speb_convergence} shows the SPEB averaged over all STs during training. The SPEB is computed as the trace of the Cram\'er--Rao bound (CRB) matrix for planar position, as derived in Appendix~\ref{app:position_crb}. In the later stages of training, SI-TokCom achieves an SPEB comparable to Centralized and SemanticCom and generally lower than Local and the two single-token variants. These results support the benefit of joint state and intent exchange for cooperative sensing, with the lower SPEB suggesting the potential for more accurate localization.

\begin{figure*}[t]
    \centering
    \subfloat[$K=4$, $N_s=3$, $N_c=8$]{%
        \includegraphics[width=0.32\textwidth]{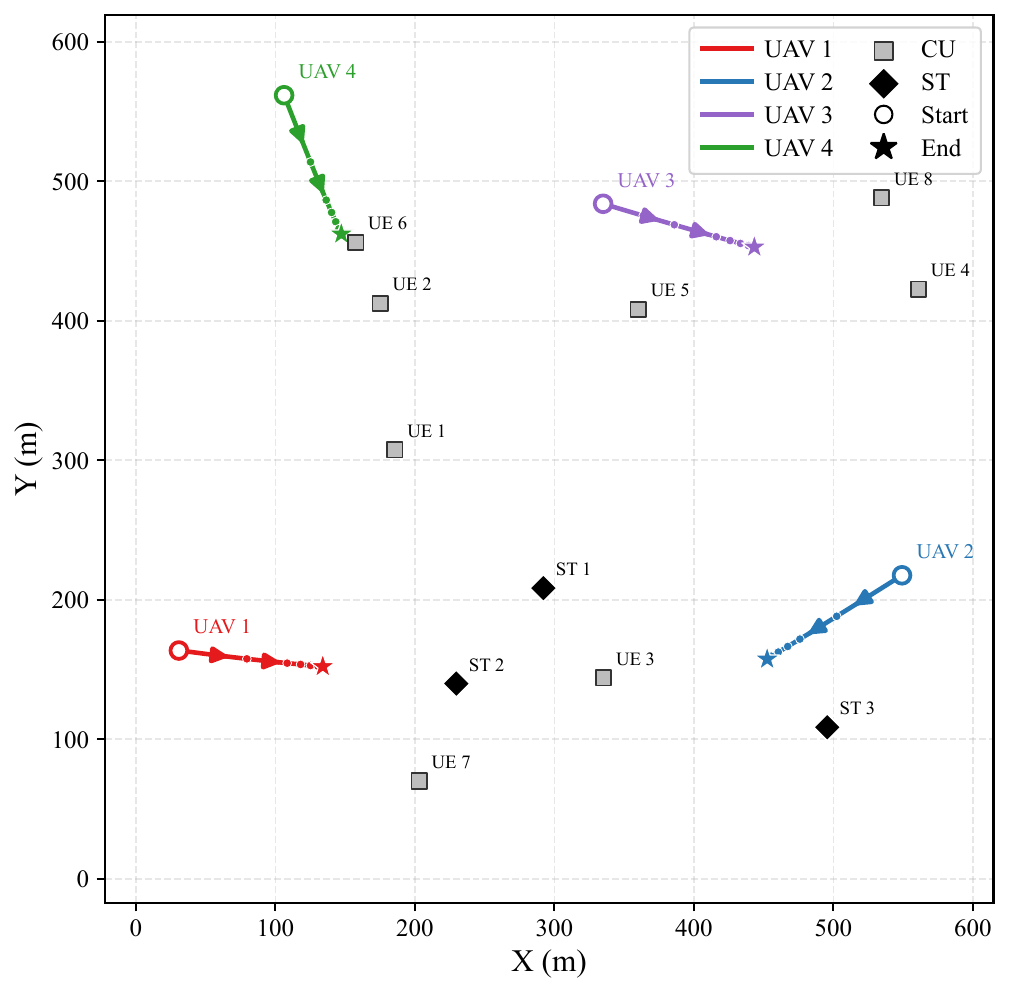}%
        \label{fig:uav_trajectory_scenario2}%
    }\hfill
    \subfloat[$K=6$, $N_s=4$, $N_c=10$]{%
        \includegraphics[width=0.32\textwidth]{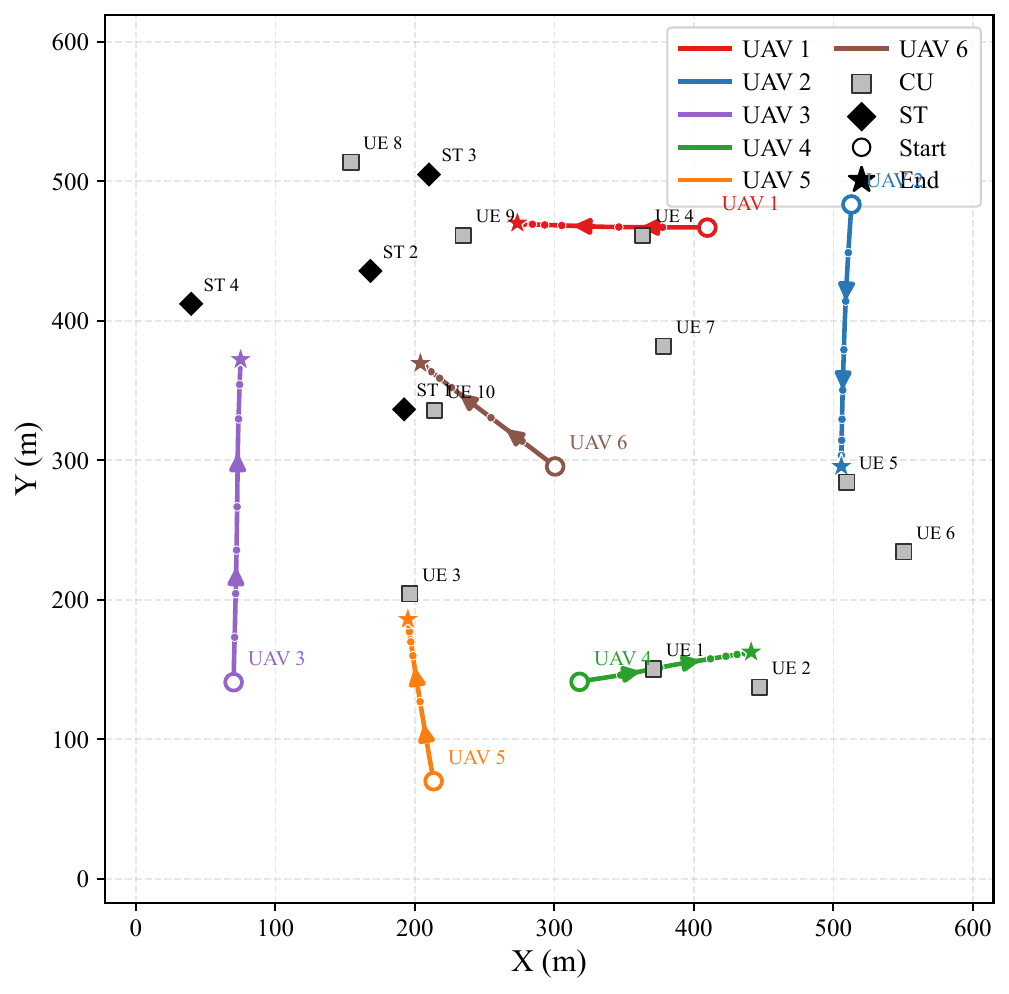}%
        \label{fig:uav_trajectory_scenario3}%
    }\hfill
    \subfloat[$K=3$, $N_s=2$, $N_c=6$]{%
        \includegraphics[width=0.32\textwidth]{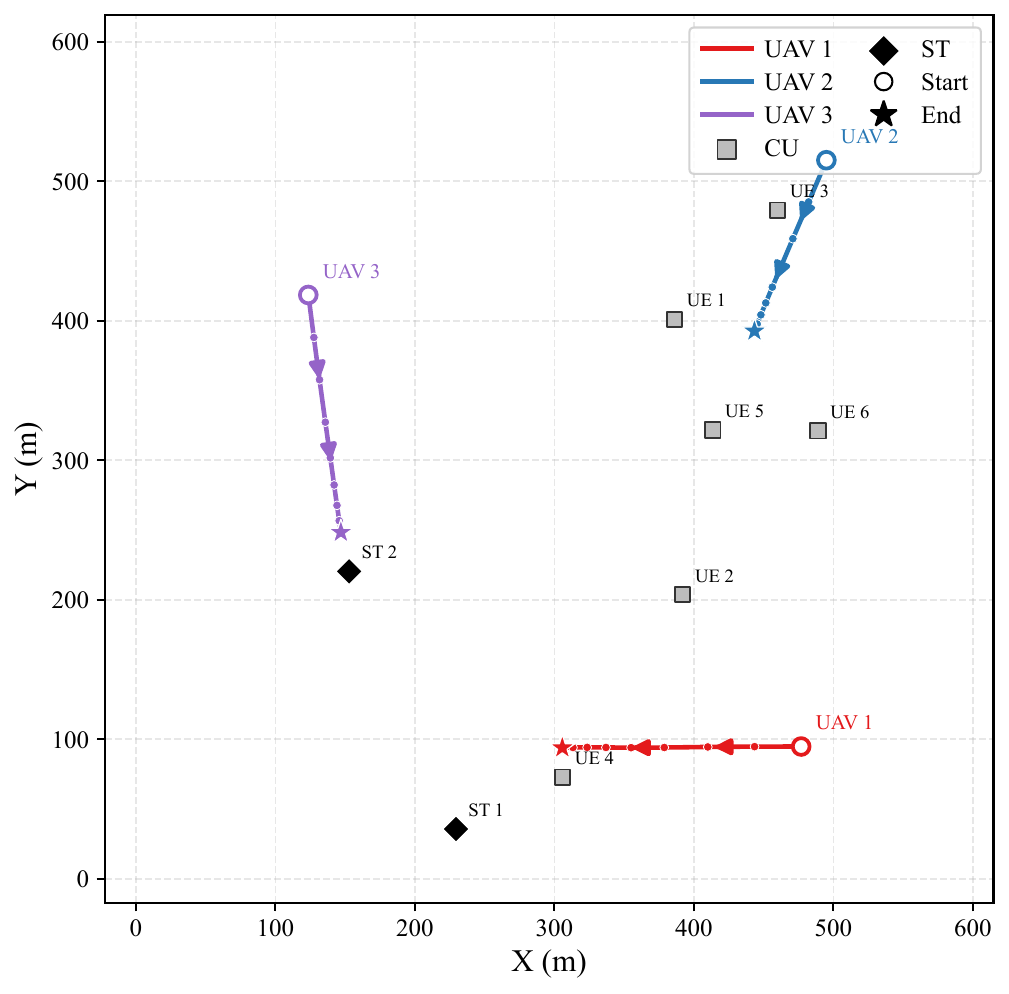}%
        \label{fig:uav_trajectory_scenario4}%
    }
    \caption{UAV trajectories generated by TokenCom under different network sizes and randomly generated CU and ST layouts.}
    \label{fig:uav_trajectories}
\end{figure*}


Fig.~\ref{fig:uav_trajectories} illustrates the UAV trajectories generated by SI-TokCom in three test scenarios with different numbers of UAVs, STs, and CUs. The ground-node locations are randomly generated and differ from those used during training. The trajectory variations provide qualitative evidence of adaptability across the evaluated layouts and network sizes. This adaptability may benefit from the complementary roles of state and intent tokens: state tokens convey local conditions, while intent tokens help UAVs anticipate their peers' planned actions. Together, these cues may support coordinated mobility as service locations and the number of UAVs change.



Fig.~\ref{fig:rate_speb_tradeoff} compares the tradeoff between the communication and sensing rates for SI-TokCom, Local, Centralized, State Only, Intent Only, and SemanticCom during online testing. 
Varying $\omega_1$ and $\omega_2$ adjusts the relative priorities of the two services and produces different operating points. Higher communication rates are generally accompanied by lower sensing rates. At comparable communication rates, SI-TokCom generally achieves higher sensing rates than Local, State Only, Intent Only, and SemanticCom, while Centralized provides the best overall tradeoff. These comparisons support the benefit of joint state and intent exchange in balancing communication and sensing performance.
\begin{figure}[t]
    \centering
    \includegraphics[
        width=\linewidth,
        trim={10mm 0mm 0mm 0mm},
        clip
    ]{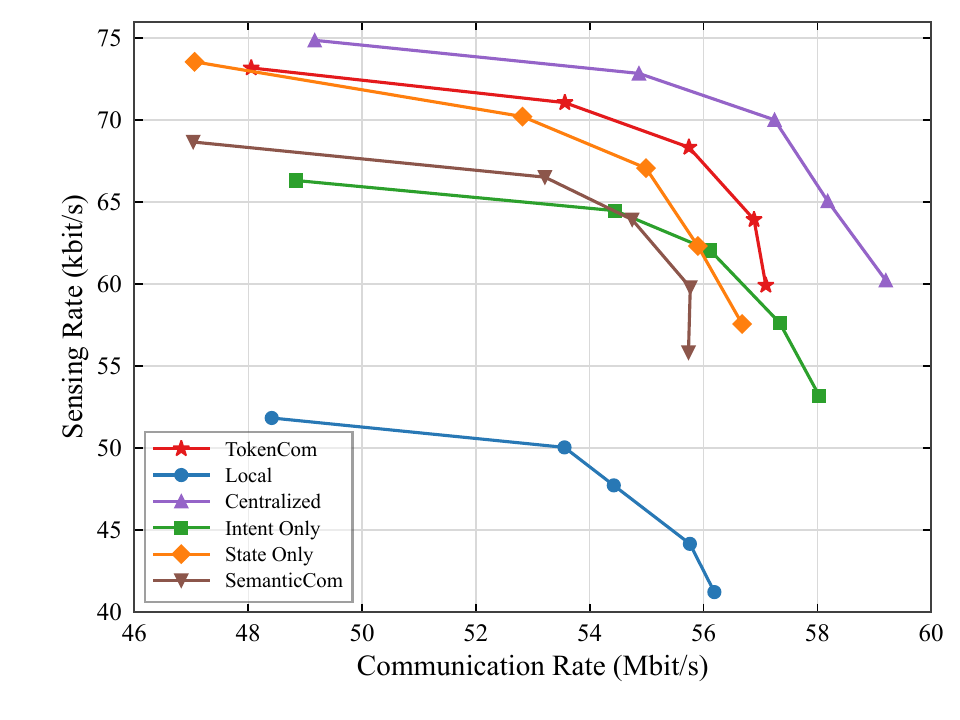}
    \caption{Tradeoff between the communication and sensing rates.}
    \label{fig:rate_speb_tradeoff}
\end{figure}



Fig.~\ref{fig:codebook_usage_difference}(a) and (b) show the S-Token and I-Token usage frequencies in scenario A, respectively, providing the reference distributions for comparison. Scenario A has higher communication and lower sensing requirements, ground nodes near the initial UAV positions, and high initial battery levels. Scenario B has lower communication and higher sensing requirements, ground nodes clustered in the southwestern region, and low initial battery levels. Fig.~\ref{fig:codebook_usage_difference}(c) and (d) present the corresponding frequency differences, computed as scenario B minus scenario A. Positive and negative values indicate increased and decreased usage in scenario B, respectively. These nonuniform shifts reveal scenario-dependent token preferences, providing qualitative evidence of adaptive codebook reuse under changing service requirements, node layouts, and energy conditions.

\begin{figure*}[t]
    \centering
    \subfloat[S-Token usage in scenario A]{%
        \includegraphics[width=0.24\textwidth]{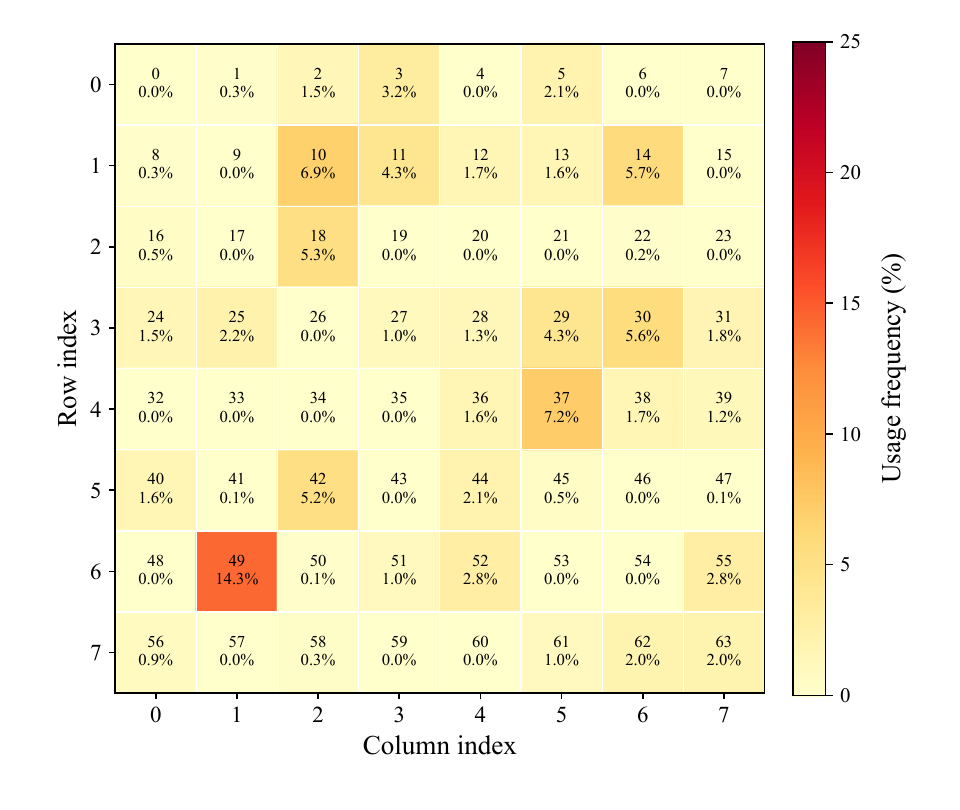}%
        \label{fig:state_codebook_usage_A}%
    }\hfill
    \subfloat[I-Token usage in scenario A]{%
        \includegraphics[width=0.24\textwidth]{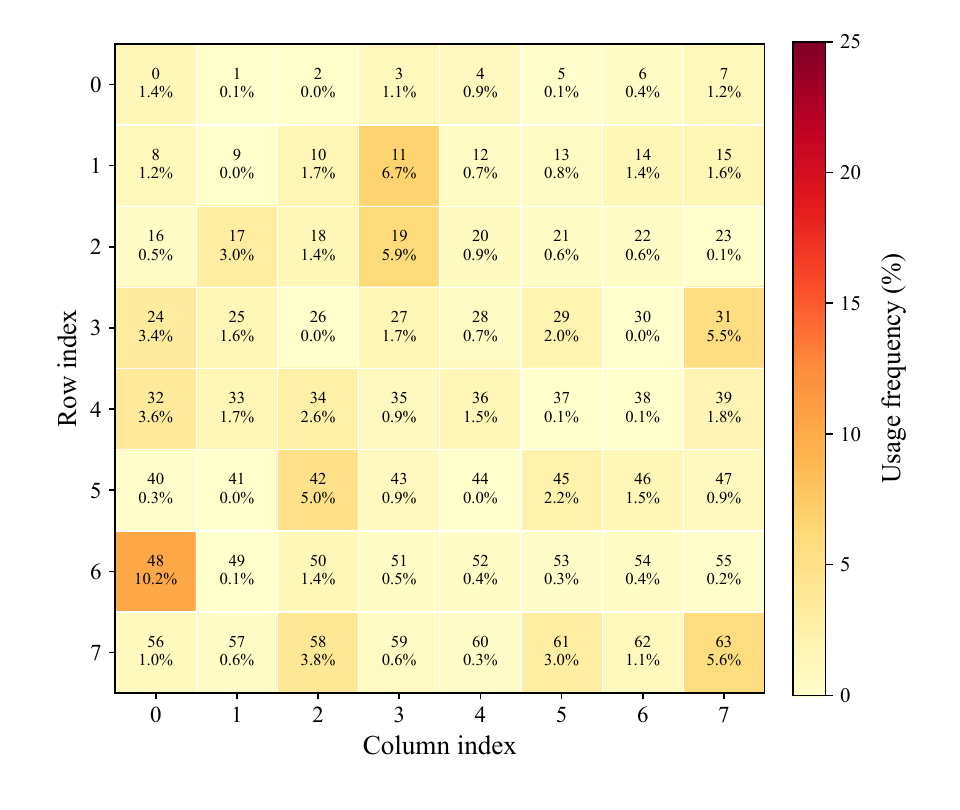}%
        \label{fig:intent_codebook_usage_A}%
    }\hfill
    \subfloat[S-Token usage difference]{%
        \includegraphics[width=0.24\textwidth]{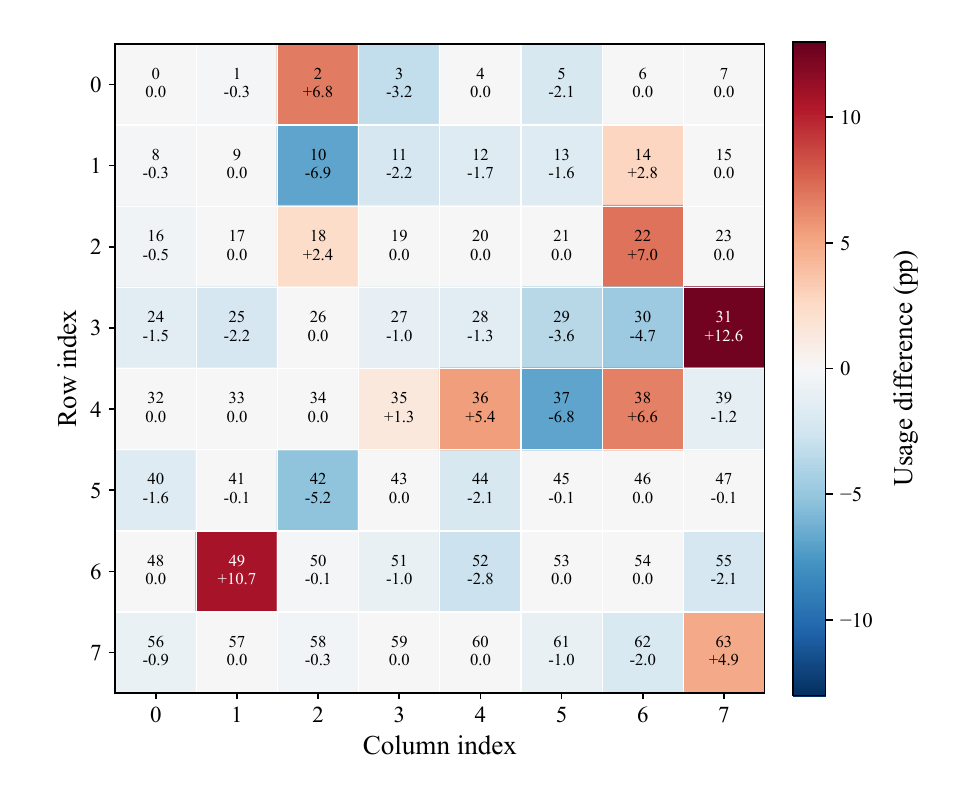}%
        \label{fig:state_codebook_usage_difference}%
    }\hfill
    \subfloat[I-Token usage difference]{%
        \includegraphics[width=0.24\textwidth]{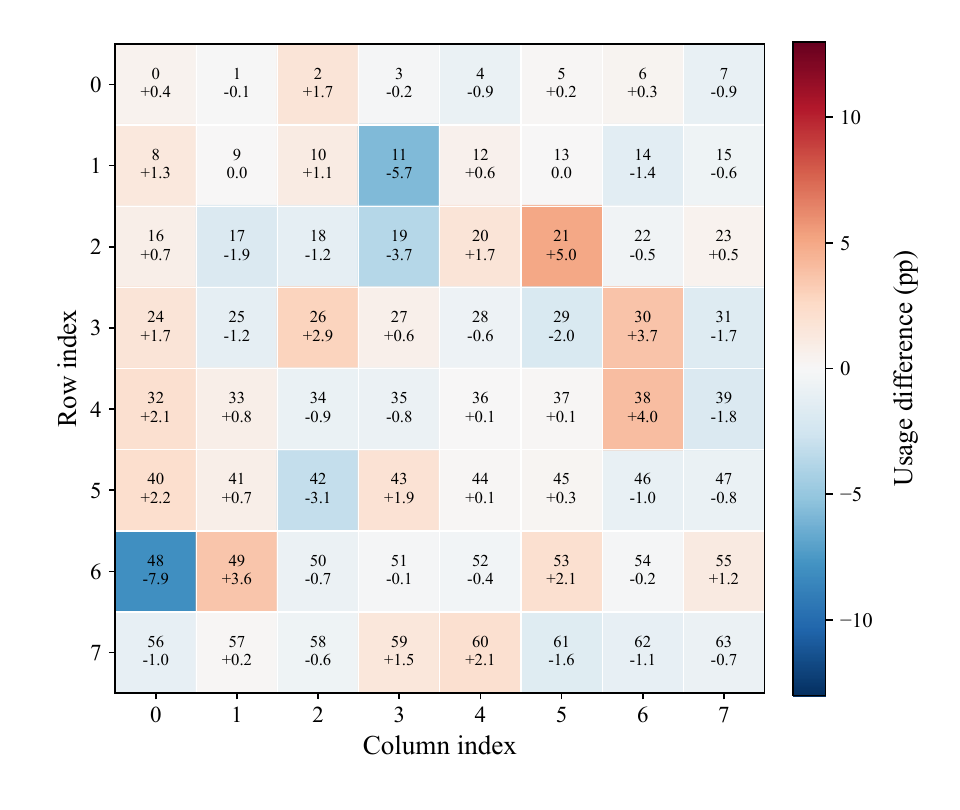}%
        \label{fig:intent_codebook_usage_difference}%
    }
    \caption{Token usage statistics: (a), (b) usage frequencies in scenario A; (c), (d) usage-frequency differences between scenarios B and A (scenario B minus scenario A). Frequencies are expressed as percentages and differences as percentage points.}
    \label{fig:codebook_usage_difference}
\end{figure*}

Fig.~\ref{fig:sensing_spatial_performance} compares the terminal SPEB in a sensing-only task with $N_s=1$ and $N_c=0$, where each UAV dedicates all its bandwidth and transmit beams to sensing. The sensing SNR therefore becomes $\gamma_{k,j}^{s}[t] = G_p \| \rho_{k,j}^{s}[t] {\mathbf{a}_{k,j}^{s}[t]}^{H} \mathbf{W}_{k}[t] \|_2^2 / N_0 B_k$, where $\mathbf{W}_{k}[t] = [ \mathbf{w}_{k,s_1}[t], \dots, \mathbf{w}_{k,s_{N_s}}[t] ]$ denotes the sensing beamforming matrix of UAV $k$.
Fig.~\ref{fig:sensing_spatial_performance}(a) shows the spatial SPEB distribution, Fig.~\ref{fig:sensing_spatial_performance}(b) presents the central slices at $x=300~\mathrm{m}$ and $y=300~\mathrm{m}$, and Fig.~\ref{fig:sensing_spatial_performance}(c) reports the empirical cumulative distribution function (CDF) over all sampled ST locations. The location-dependent peaks in (b) highlight the influence of sensing geometry on localization performance. SI-TokCom generally achieves lower SPEB than Local and can outperform Centralized along both slices. Its generally higher CDF further indicates that a larger fraction of sampled locations satisfies a given SPEB threshold, suggesting that the gains extend beyond the selected slices. These gains may stem from scenario-relevant token selection and explicit intent exchange, which help anticipate other UAVs' actions and improve sensing coordination. Although Centralized has access to global information, SI-TokCom's intent prediction may offer a more pronounced advantage in this simpler scenario.

\begin{figure*}[t]
    \centering
    \subfloat[Spatial SPEB distribution]{%
        \includegraphics[width=0.32\textwidth]{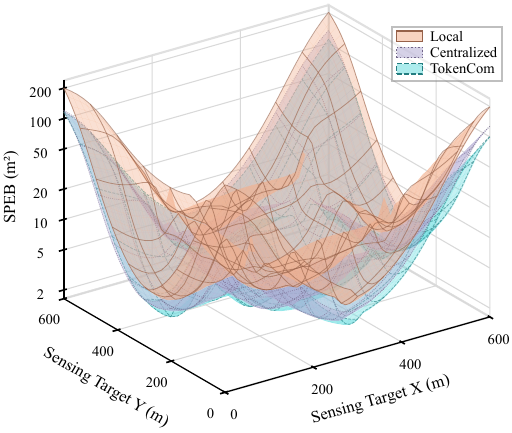}%
        \label{fig:sensing_speb_surface}%
    }\hfill
    \subfloat[Central SPEB slices]{%
        \includegraphics[width=0.32\textwidth]{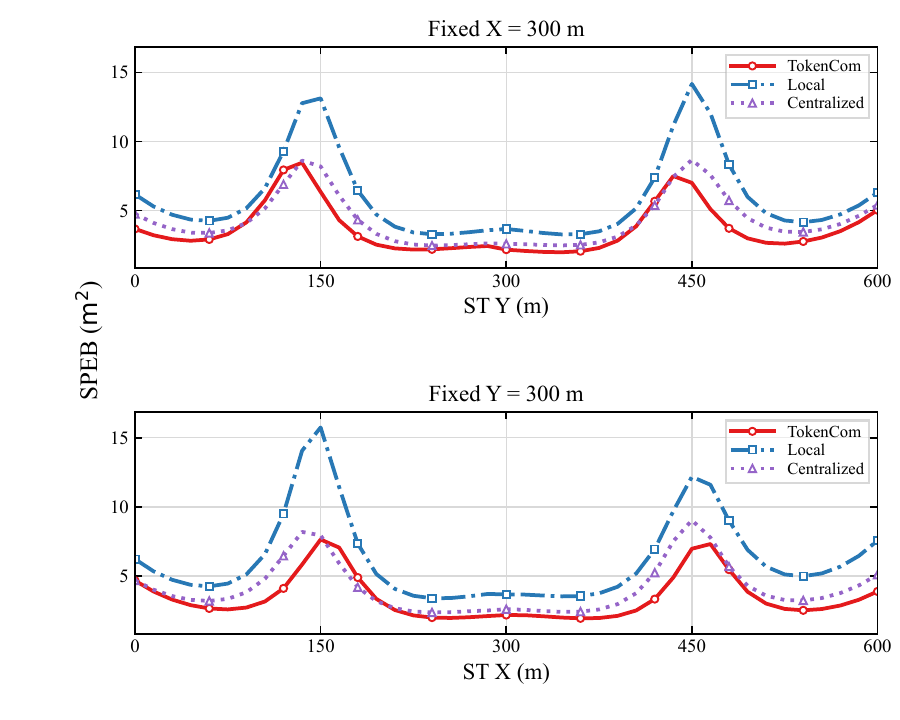}%
        \label{fig:sensing_speb_slice_x_y}%
    }\hfill
    \subfloat[Empirical CDF of SPEB]{%
        \includegraphics[width=0.32\textwidth]{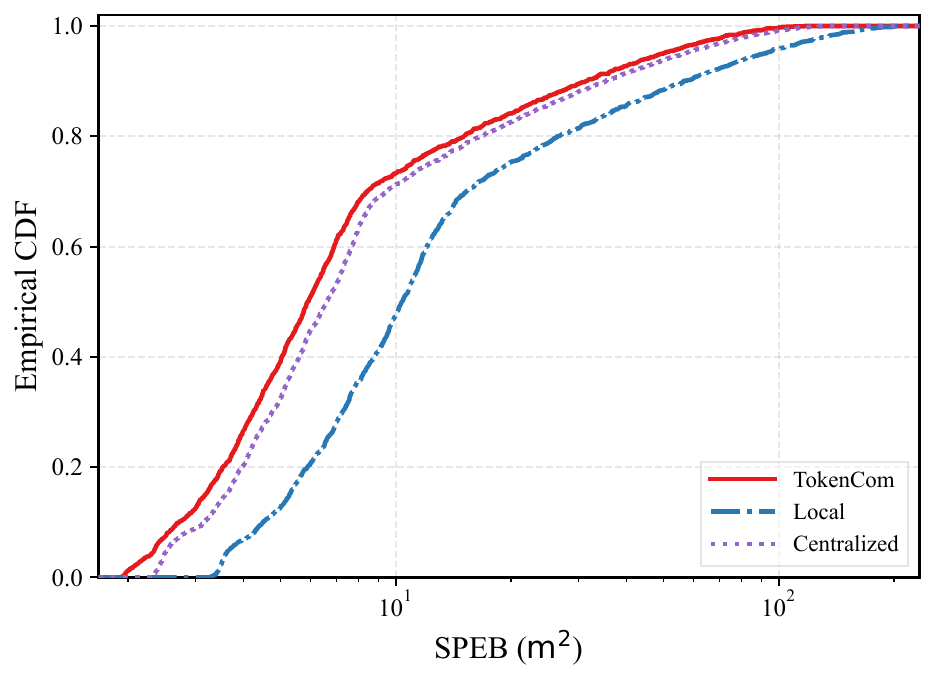}%
        \label{fig:spatial_cdf_compare}%
    }
    \caption{Terminal SPEB at $t=12$ in a sensing-only test with $N_s=1$ and $N_c=0$: (a) spatial distribution over the ST position grid; (b) central slices at $x=300~\mathrm{m}$ (upper) and $y=300~\mathrm{m}$ (lower); and (c) empirical CDF over all sampled ST locations.}
    \label{fig:sensing_spatial_performance}
\end{figure*}

\section{Conclusion}
\label{sec:conclusion}

This paper developed SI-TokCom, an embodied multi-agent framework for cooperative multi-UAV ISAC. We formulated a joint TokCom and physical control problem to minimize propulsion energy subject to communication and sensing requirements. The proposed alternating optimization algorithm combined MAPPO-based joint learning of token selection, composition, and physical control with closed-form beamforming updates. Simulation results showed that SI-TokCom achieved 98.9\% and 99.1\% of the communication and sensing rates of Centralized, respectively, with a coordination payload of 96 bits per UAV per slot. Compared with Local, SI-TokCom improved these rates by 5.0\% and 43.6\%, respectively, with comparable propulsion energy. Changes in token usage across the evaluated scenarios further indicated that fixed semantic codebooks could support adaptive coordination. These findings suggest that jointly learning compact messages and physical decisions can enable effective decentralized cooperation in multi-UAV ISAC.

\begin{appendices}

\section{Derivation of the Closed-Form Beamforming Update}
\label{app:fp_beamforming}
First, we rewrite the objective function in \eqref{problem:SP1_obj} based on the Lagrangian dual transform and quadratic transform under the FP framework~\cite{FP}, which can be expressed as

\begin{equation}
\label{eq:FP_transformed_objective}
    \begin{aligned}
        \widetilde{G}\left(\mathbf{W}_{k,i}[t]\right)
    &= \sum_{r=1}^{N_s+1}
    \Bigl(
    2\operatorname{Re}
    \left\{
    \boldsymbol{\varphi}_r^{H}[t]
    \mathbf{w}_{k,i,r}[t]
    \right\}
    \\
    &\qquad -
    \mathbf{w}_{k,i,r}^{H}[t]
    \boldsymbol{\Lambda}_r[t]
    \mathbf{w}_{k,i,r}[t]
    \Bigr)
    +
    B[t],
    \end{aligned}
\end{equation}
where $\boldsymbol{\varphi}_r[t]$, $\boldsymbol{\Lambda}_r[t]$, and $B[t]$ are defined in \eqref{eq:FP_phi}, \eqref{eq:FP_Lambda}, and \eqref{eq:FP_B}, respectively.
The Lagrangian dual transform introduces the auxiliary variables $\mu_j^{s}$ and $\mu^{c}$, whose optimal updates are given by
\stepcounter{equation}\edef 
\FPBnumber 
{\theequation} 
\begin{equation}
\mu_j^{s\star}
=
\frac{
\alpha_{k,i}[t]G_p
\left\|
\rho_{k,j}^{s}
(\mathbf{a}_{k,j}^{s})^H
\mathbf{W}_{k,i}[t]
\right\|_2^2
} { N_0 b_{k,i}[t]},
\label{eq:mu_s_opt}
\end{equation}
\begin{equation}
\mu^{c\star}
=
\frac{
\left|
(\mathbf h_{k,i}^{c})^H\mathbf w_{k,i}
\right|^2
}{
\sum_{r=1}^{N_s}
\left|
(\mathbf h_{k,i}^{c})^H\mathbf w_{k,i,r}
\right|^2
+N_0b_{k,i}[t]
}.
\label{eq:mu_c_opt}
\end{equation}
The quadratic transform further decouples the remaining ratios by introducing $\boldsymbol{\xi}_j^{s}$ and $\xi^{c}$, whose optimal updates are given by 
\begin{equation}
\boldsymbol{\xi}_j^{s\star} = \frac{ \alpha_{k,i}[t]\sqrt{1+\mu_j^{s}}\sqrt{G_p} \rho_{k,j}^{s} (\mathbf{a}_{k,j}^{s})^H \mathbf{W}_{k,i}
}{ \alpha_{k,i}[t]G_p
\left\|
\rho_{k,j}^{s}
(\mathbf{a}_{k,j}^{s})^H
\mathbf{W}_{k,i}
\right\|_2^2 + N_0b_{k,i}[t]
},
\label{eq:xi_s_opt}
\end{equation}
\begin{equation}
\xi^{c\star}
=
\frac{
\sqrt{1+\mu^{c}}
(\mathbf h_{k,i}^{c})^H\mathbf w_{k,i,1}
}{
\sum_{r=1}^{N_s+1}
\left|
(\mathbf h_{k,i}^{c})^H\mathbf w_{k,i,r}
\right|^2
+N_0b_{k,i}[t]
},
\label{eq:xi_c_opt}
\end{equation}
where
$\boldsymbol{\xi}_j^{s}
=[\xi_{j,1}^{s},\ldots,\xi_{j,N_s+1}^{s}]$ is a row vector. 



For fixed auxiliary variables, $\boldsymbol{\Lambda}_r\succeq\mathbf0$, and the resulting beamforming subproblem is convex. By omitting the constant term $B$, the Lagrangian associated with the power constraint is expressed as
\begin{equation}
\begin{aligned}
\mathcal{L}\left(\mathbf{W}_{k,i},\lambda\right)
=&
\sum_{r=1}^{N_s+1}
\left(
\mathbf{w}_{k,i,r}^{H}
\boldsymbol{\Lambda}_r
\mathbf{w}_{k,i,r}
-
2\operatorname{Re}
\left\{
\boldsymbol{\varphi}_r^{H}
\mathbf{w}_{k,i,r}
\right\}
\right)
\\
&+
\lambda
\left(
\sum_{r=1}^{N_s+1}
\left\|
\mathbf{w}_{k,i,r}
\right\|_2^2
-
p_{k,i}[t]
\right),
\end{aligned}
\label{eq:beamforming_Lagrangian}
\end{equation}
where $\lambda\geq0$ is the dual variable associated with the power constraint. For $p_{k,i}[t]>0$, Slater's condition holds, so the KKT conditions are necessary and sufficient. Stationarity requires
\begin{equation}
\frac{\partial\mathcal L}{\partial\mathbf w_{k,i,r}^*}
=
\left( \boldsymbol{\Lambda}_r+\lambda\mathbf I_M \right)\mathbf w_{k,i,r}
-\boldsymbol{\varphi}_r = \mathbf0, 
\label{eq:beamforming_KKT}
\end{equation}
which yields \eqref{eq:beamforming_closed_form}.
The optimal multiplier follows from complementary slackness:
\begin{equation}
\lambda^{\star}
\left(
\sum_{r=1}^{N_s+1}
\|\mathbf w_{k,i,r}^{\star}\|_2^2-p_{k,i}[t]
\right)
=0.
\label{eq:app_complementary_slackness}
\end{equation}
For $\lambda>0$, define
\begin{equation}
h(\lambda)
=
\sum_{r=1}^{N_s+1}
\left\|
\left(
\boldsymbol{\Lambda}_r+\lambda\mathbf I_M
\right)^{-1}
\boldsymbol{\varphi}_r
\right\|_2^2
-p_{k,i}[t],
\label{eq:dual_function_lambda}
\end{equation}
with
$h(0)=\sum_{r=1}^{N_s+1}
\|\boldsymbol{\Lambda}_r^{\dagger}\boldsymbol{\varphi}_r\|_2^2
-p_{k,i}[t]$.
If $h(0)\leq0$, then $\lambda^{\star}=0$. Otherwise, the power constraint is active, and $\lambda^{\star}>0$ satisfies $h(\lambda^{\star})=0$. Differentiating \eqref{eq:dual_function_lambda} gives
\begin{equation}
h'(\lambda)
=
-2\sum_{r=1}^{N_s+1}
\boldsymbol{\varphi}_r^H
\left(
\boldsymbol{\Lambda}_r+\lambda\mathbf I_M
\right)^{-3}
\boldsymbol{\varphi}_r
\leq0.
\label{eq:app_dual_monotonicity}
\end{equation}
The inequality is strict when $h(0)>0$. Since $\lim_{\lambda\to\infty}h(\lambda)=-p_{k,i}[t]<0$, the positive root is unique and can be efficiently obtained via bisection.
This completes the proof of Theorem~\ref{thm:beamforming_update}.

\section{Derivation of the CRB for Planar Position}
\label{app:position_crb}

The distance measurements obtained by the $K$ UAVs at time slot $t$
are collected as 
\begin{equation}
\widetilde{\mathbf d}_j^{s}[t]
=
\mathbf d_j^{s}[t]
+
\mathbf e_j^{s}[t]
\sim
\mathcal{N}
\left(
\mathbf d_j^{s}[t],
\mathbf C_j^{s}[t]
\right),
\label{eq:app_range_measurement}
\end{equation}
where
$\mathbf d_j^{s}[t]
=[d_{1,j}^{s}[t],\ldots,d_{K,j}^{s}[t]]^{T}$ and
$ \mathbf C_j^{s}[t] = \operatorname{diag}\left( \left(\sigma_{1,j}^{d}[t]\right)^2,
\ldots,\left(\sigma_{K,j}^{d}[t]\right)^2 \right)$.
The individual ranging error is modeled as
\begin{equation}
e_{k,j}^{s}[t]
\sim
\mathcal{N}
\left(
0,
\left(\sigma_{k,j}^{d}[t]\right)^2
\right),
\label{eq:app_ranging_error}
\end{equation}
where $\left(\sigma_{k,j}^{d}[t]\right)^2
=\zeta_s/\gamma_{k,j}^{s}[t]$
and $\zeta_s$ denotes the ranging-error coefficient.
By aggregating the sensing information carried by the allocated
subchannels, the bandwidth-weighted sensing SNR is defined as
$\gamma_{k,j}^{s}[t]
=\sum_{i=1}^{N_c}
\alpha_{k,i}[t]\gamma_{k,j,i}^{s}[t]
\left(b_{k,i}[t]/b_0\right)^3$,
where $b_0$ is the reference bandwidth used for normalization.

Based on the above Gaussian ranging model, the FIM $\mathbf J_j^{d}[t]$ for the distance parameters is given by \eqref{eq:app_distance_fim_general}, where
$\chi_{k,j}^{s}[t] \triangleq
\sum_{i=1}^{N_c} \frac{\alpha_{k,i}[t]b_{k,i}[t]^2}{b_0^3} \left\| (\mathbf a_{k,j}^{s}[t])^H \mathbf W_{k,i}[t] \right\|_2^2$ 
and $\mathbf D=\operatorname{diag}(0,1,\ldots,M-1)$.



\begin{figure*}[t]

\begin{equation*}
\tag{\FPBnumber}  
\begin{aligned}
B[t]
=&
\omega_1 b_{k,i}[t]
\sum_{j=1}^{N_s}
\log\left(1+\mu_j^{s}\right)
+
\omega_2 b_{k,i}[t]
\log\left(1+\mu^{c}\right)
-
\omega_1 b_{k,i}[t]\mu_j^{s}
-
\omega_2 b_{k,i}[t]\mu^{c}
\\
&-
\omega_2 b_{k,i}[t]
\left|\xi^{c}\right|^2
N_0 b_{k,i}[t]
-
\omega_1 b_{k,i}[t]
\left|\xi_j^{s}\right|^2
N_0 b_{k,i}[t].
\end{aligned}
\label{eq:FP_B}
\end{equation*}
\vspace{-0.7cm}

\begin{equation}
\begin{aligned}
&\left[\mathbf J_j^{d}[t]\right]_{p,q}
=
\left[
\frac{\partial\mathbf d_j^{s}[t]}{\partial d_{p,j}^{s}[t]}
\right]^{T}
\left(\mathbf C_j^{s}[t]\right)^{-1}
\left[
\frac{
\partial\mathbf d_j^{s}[t]
}{
\partial d_{q,j}^{s}[t]
}
\right]
+
\frac{1}{2}
\operatorname{tr}
\left[
\left(\mathbf C_j^{s}[t]\right)^{-1}
\frac{
\partial\mathbf C_j^{s}[t]
}{
\partial d_{p,j}^{s}[t]
}
\left(\mathbf C_j^{s}[t]\right)^{-1}
\frac{
\partial\mathbf C_j^{s}[t]
}{
\partial d_{q,j}^{s}[t]
}
\right] \\
& = \frac{G_p \beta_{\rm s}  }{ \zeta_s  N_0}  \frac{ \chi_{k,j}^{s}[t] }{ \left(d_{k,j}^{s}[t]\right)^{4} } + 
\frac{1}{2} \left( -\frac{4 \pi d_a H}{ \lambda \chi_{k,j}^{s}[t] \left( d_{k,j}^{s} \right)^{2} } \sum_{i=1}^{N_c}\frac{\alpha_{k,i}[t]b_{k,i}[t]^2}{b_0^3}  
\!\operatorname{Im} \left\{ \left( (\mathbf a_{k,j}^{s}[t])^H \mathbf W_{k,i}[t] \right)^{*} (\mathbf a_{k,j}^{s}[t])^H  \mathbf{D}\mathbf{W}_{k,i}[t] \right\} + \frac{4}{ d_{k,j}^{s}[t]} \right)^2, \\
& \quad p=q=k, k=1,\ldots,K.
\end{aligned}
\label{eq:app_distance_fim_general}
\end{equation}


\hrulefill
\end{figure*}

Accordingly, the FIM for the planar target position is obtained through the Jacobian transformation as
\begin{equation}
\mathbf J_j^{p}[t]
=
\mathbf Q_j[t]
\mathbf J_j^{d}[t]
\mathbf Q_j^{T}[t],
\label{eq:app_position_fim_transform}
\end{equation}
where $\mathbf Q_j[t]=[\mathbf q_{1,j}[t],\ldots,\mathbf q_{K,j}[t]]$
is the Jacobian matrix, whose $k$th column is
\begin{equation}
\mathbf q_{k,j}[t]
=
\begin{bmatrix} 
\frac{\partial d_{k,j}^{s}[t]}{\partial x_j^{s}}\\ 
\frac{\partial d_{k,j}^{s}[t]}{\partial y_j^{s}} 
\end{bmatrix}
=
\frac{1}{d_{k,j}^{s}[t]}
\begin{bmatrix}
\Delta x_{k,j}[t]\\
\Delta y_{k,j}[t]
\end{bmatrix},
\label{eq:app_distance_gradient}
\end{equation}
where $\Delta x_{k,j}[t]=x_j^{s}-x_k[t]$ and $\Delta y_{k,j}[t]=y_j^{s}-y_k[t]$.
Substituting \eqref{eq:app_distance_gradient} into
\eqref{eq:app_position_fim_transform} yields
\begin{equation}
\begin{aligned}
\mathbf J_j^{p}[t]
=
\sum_{k=1}^{K}
J_{j,k}^{d}[t]
\mathbf q_{k,j}[t]\mathbf q_{k,j}^{T}[t]
=
\begin{bmatrix}
\Theta_j^{a}[t] & \Theta_j^{c}[t]\\
\Theta_j^{c}[t] & \Theta_j^{b}[t]
\end{bmatrix},
\end{aligned}
\label{eq:app_position_fim_matrix}
\end{equation}
where $\Theta_j^{a}[t]$, $\Theta_j^{b}[t]$, and $\Theta_j^{c}[t]$
are given by 
\begin{equation}
\Theta_j^{a}[t]
=
\sum_{k=1}^{K}
\left(
\frac{\Delta x_{k,j}[t]}{d_{k,j}^{s}[t]}
\right)^2
\left[\mathbf J_j^{d}[t]\right]_{k,k},
\label{eq:app_theta_a}
\end{equation}
\begin{equation}
\Theta_j^{b}[t]
=
\sum_{k=1}^{K}
\left(
\frac{\Delta y_{k,j}[t]}{d_{k,j}^{s}[t]}
\right)^2
\left[\mathbf J_j^{d}[t]\right]_{k,k},
\label{eq:app_theta_b}
\end{equation}

\begin{equation}
\Theta_j^{c}[t]
=
\sum_{k=1}^{K}
\frac{\Delta x_{k,j}[t]\Delta y_{k,j}[t]}
{\left(d_{k,j}^{s}[t]\right)^2}
\left[\mathbf J_j^{d}[t]\right]_{k,k},
\label{eq:app_theta_c}
\end{equation}
where $\left[\mathbf J_j^{d}[t]\right]_{k,k}$ is given by
\eqref{eq:app_distance_fim_general} with $p=q=k$.
Substituting \eqref{eq:app_position_fim_matrix} into the definition of the planar position CRB yields
\begin{equation}
\begin{aligned}
\mathrm{CRB}_j[t]
&=
\left(\mathbf J_j^{p}[t]\right)^{-1}
\\
&=
\frac{1}{
\Theta_j^{a}[t]\Theta_j^{b}[t]
-
\left(\Theta_j^{c}[t]\right)^2
}
\begin{bmatrix}
\Theta_j^{b}[t] & -\Theta_j^{c}[t]\\
-\Theta_j^{c}[t] & \Theta_j^{a}[t]
\end{bmatrix}.
\end{aligned}
\label{eq:app_position_crb}
\end{equation}

Finally, taking the trace of the CRB matrix in \eqref{eq:app_position_crb} yields the SPEB for ST $j$
\begin{equation}
\mathrm{SPEB}_j[t]
=
\operatorname{tr}\left(\mathrm{CRB}_j[t]\right)
=
\frac{
\Theta_j^{a}[t]+\Theta_j^{b}[t]
}{
\Theta_j^{a}[t]\Theta_j^{b}[t]
-
\left(\Theta_j^{c}[t]\right)^2
}.
\label{eq:app_speb}
\end{equation}

\end{appendices}

\bibliographystyle{IEEEtran}
\bibliography{ref}

@IEEEtranBSTCTL{IEEEauthorControl,
  CTLuse_forced_etal       = "yes",
  CTLmax_names_forced_etal = "6",
  CTLnames_show_etal       = "1"
}

@ARTICLE{letaief2021edge,
  author={Letaief, Khaled B. and Shi, Yuanming and Lu, Jianmin and Lu, Jianhua},
  journal={IEEE J. Sel. Areas Commun.},
  title={Edge artificial intelligence for {6G}: Vision, enabling technologies, and applications},
  year={2022},
  month={Jan.},
  volume={40},
  number={1},
  pages={5--36},
  doi={10.1109/JSAC.2021.3126076}
}

@ARTICLE{Liu2022ISAC,
  author={Liu, Fan and Cui, Yuanhao and Masouros, Christos and Xu, Jie and Han, Tony Xiao and Eldar, Yonina C. and Buzzi, Stefano},
  journal={IEEE J. Sel. Areas Commun.},
  title={Integrated Sensing and Communications: Toward Dual-Functional Wireless Networks for {6G} and Beyond},
  year={2022},
  month={Jun.},
  volume={40},
  number={6},
  pages={1728--1767},
  doi={10.1109/JSAC.2022.3156632}
}

@ARTICLE{xiu2025movable,
  author={Xiu, Yue and Yang, Songjie and Lyu, Wanting and Yeoh, Phee Lep and Li, Yonghui and Ai, Yi},
  journal={IEEE Wireless Commun. Lett.},
  title={Movable antenna enabled {ISAC} beamforming design for low-altitude airborne vehicles},
  year={2025},
  month={May},
  volume={14},
  number={5},
  pages={1311--1315},
  doi={10.1109/LWC.2025.3538587}
}

@ARTICLE{yang2026embodied,
  author={Yang, Helin and Du, Honglin and Geng, Qing and Xu, Changyuan and Xiong, Zehui},
  journal={IEEE Trans. Wireless Commun.},
  title={Embodied intelligence-enhanced anti-jamming resource allocation for low-altitude communication networks},
  year={2026},
  volume={25},
  pages={14750--14765},
  doi={10.1109/TWC.2026.3677272}
}

@MISC{ma2026semantic,
  author={Ma, Yu and Gao, Zhen and Qiao, Li and Zhang, Xiaoyuan and Mashhadi, Mahdi Boloursaz and Xu, Yin and Xu, Wenjun and Xu, Xiaodong and Huang, Kaibin and Wang, Jiangzhou and Tafazolli, Rahim and Chen, Sheng and Quek, Tony Q. S. and Zhang, Ping},
  title={From Semantic to Token Communication: The Next Paradigm for Large-Model-Driven {6G} Intelligent Connectivity},
  year={2026},
  note={arXiv:2609.10714}
}

@ARTICLE{xiu2024improving,
  author={Xiu, Yue and Lyu, Wanting and Yeoh, Phee Lep and Ai, Yi and Li, Yonghui and Wei, Ning},
  journal={IEEE Trans. Veh. Technol.},
  title={Improving physical-layer security in {ISAC-AAV} system: Beamforming and trajectory optimization},
  year={2025},
  month={Feb.},
  volume={74},
  number={2},
  pages={3503--3508},
  doi={10.1109/TVT.2024.3477925}
}

@ARTICLE{Pan2024CooperativeISAC,
  author={Pan, Yu and Li, Ruoguang and Da, Xinyu and Hu, Hang and Zhang, Miao and Zhai, Dong and Cumanan, Kanapathippillai and Dobre, Octavia A.},
  journal={IEEE Trans. Veh. Technol.},
  title={Cooperative Trajectory Planning and Resource Allocation for {UAV}-Enabled Integrated Sensing and Communication Systems},
  year={2024},
  month={May},
  volume={73},
  number={5},
  pages={6502--6516},
  doi={10.1109/TVT.2023.3337106}
}

@ARTICLE{Zhou2026JRC,
  author={Zhou, Lingyun and Yang, Chunyong and Cui, Yongqiang and Zhang, Rongqing and Wei, Zhongxiang and Shi, Qingjiang},
  journal={IEEE Trans. Wireless Commun.},
  title={Joint Deployment and Resource Allocation Design for {JRC}-Enabled Multi-{UAV} Cooperative Systems},
  year={2026},
  volume={25},
  pages={8051--8065},
  doi={10.1109/TWC.2025.3635277}
}

@ARTICLE{Qin2023DRLISAC,
  author={Qin, Yunhui and Zhang, Zhongshan and Li, Xulong and Huangfu, Wei and Zhang, Haijun},
  journal={IEEE Trans. Wireless Commun.},
  title={Deep Reinforcement Learning Based Resource Allocation and Trajectory Planning in Integrated Sensing and Communications {UAV} Network},
  year={2023},
  month={Nov.},
  volume={22},
  number={11},
  pages={8158--8169},
  doi={10.1109/TWC.2023.3260304}
}

@ARTICLE{Xie2024DistributedISAC,
  author={Xie, Zhanyuan and Wang, Ziyuan and Zhang, Zekai and Wang, Jingjing and Jiang, Zheng and Han, Zhu},
  journal={IEEE Trans. Veh. Technol.},
  title={Distributed {UAV} Swarm for Device-Free Integrated Sensing and Communication Relying on Multi-Agent Reinforcement Learning},
  year={2024},
  month={Dec.},
  volume={73},
  number={12},
  pages={19925--19930},
  doi={10.1109/TVT.2024.3438854}
}

@ARTICLE{Peng2025ISCC,
  author={Peng, Sicong and Li, Bin and Liu, Lei and Fei, Zesong and Niyato, Dusit},
  journal={IEEE Trans. Commun.},
  title={Trajectory Design and Resource Allocation for Multi-{UAV}-Assisted Sensing, Communication, and Edge Computing Integration},
  year={2025},
  month={Apr.},
  volume={73},
  number={4},
  pages={2847--2861},
  doi={10.1109/TCOMM.2024.3478115}
}

@ARTICLE{Liu2025ISCC,
  author={Liu, Yaxi and Mao, Wencan and Li, Xulong and Huangfu, Wei and Ji, Yusheng and Xiao, Yu and Long, Keping},
  journal={IEEE Trans. Commun.},
  title={Joint Task Scheduling and Resource Allocation for {UAV}-Assisted Air-Ground Collaborative Integrated Sensing, Computation, and Communication},
  year={2025},
  month={Nov.},
  volume={73},
  number={11},
  pages={10929--10945},
  doi={10.1109/TCOMM.2025.3591748}
}

@INPROCEEDINGS{deng2025semantic,
  author={Deng, Yundi and Leng, Supeng and Zhang, Ke and He, Jianhua and Zhou, Longyu and Quek, Tony},
  title={Semantic communication assisted cooperative sensing in {UAV} networks},
  booktitle={Proc. 10th Int. Conf. Comput. Commun. Syst. (ICCCS)},
  year={2025},
  month={Apr.},
  pages={862--867},
  doi={10.1109/ICCCS65393.2025.11069451}
}

@ARTICLE{Hu2025SemanticUAV,
  author={Hu, Han and Zhu, Xingwu and Zhou, Fuhui and Wu, Wei and Hu, Rose Qingyang and Zhu, Hongbo},
  journal={IEEE Trans. Commun.},
  title={Resource Allocation for Multi-Modal Semantic Communication in {UAV} Collaborative Networks},
  year={2025},
  month={Sep.},
  volume={73},
  number={9},
  pages={7599--7616},
  doi={10.1109/TCOMM.2025.3552303}
}

@ARTICLE{Guo2026PerceptionISAC,
  author={Guo, Ziji and Liu, Danpu and Zhang, Zhilong},
  journal={IEEE Trans. Veh. Technol.},
  title={Perception Enhanced Multimodal Multitask Semantic Communication and Resource Management for {UAV}-Assisted {ISAC} Systems},
  year={2026},
  month={Jun.},
  volume={75},
  number={6},
  pages={11051--11065},
  doi={10.1109/TVT.2025.3650639}
}

@ARTICLE{guo2025flexible,
  author={Guo, Jiahe and Yang, Songjie and Dong, Xiangyu and Yang, Jiapan and Deng, Junfeng and Zhang, Zhongpei and Yuen, Chau},
  journal={IEEE Internet Things J.},
  title={Flexible cylindrical arrays with movable antennas for {MISO} system: Beamforming and position optimization},
  year={2025},
  month={Sep.},
  volume={12},
  number={17},
  pages={35394--35405},
  doi={10.1109/JIOT.2025.3580372}
}

@ARTICLE{lyu2025movable,
  author={Lyu, Wanting and Yang, Songjie and Xiu, Yue and Zhang, Zhongpei and Assi, Chadi and Yuen, Chau},
  journal={IEEE Trans. Wireless Commun.},
  title={Movable antenna enabled integrated sensing and communication},
  year={2025},
  month={Apr.},
  volume={24},
  number={4},
  pages={2862--2875},
  doi={10.1109/TWC.2025.3525631}
}

@ARTICLE{FP,
  author={Shen, Kaiming and Yu, Wei},
  journal={IEEE Trans. Signal Process.},
  title={Fractional Programming for Communication Systems---Part {I}: Power Control and Beamforming},
  year={2018},
  month={May},
  volume={66},
  number={10},
  pages={2616--2630},
  doi={10.1109/TSP.2018.2812733}
}

@ARTICLE{qiao2025token,
  author={Qiao, Li and Mashhadi, Mahdi Boloursaz and Gao, Zhen and Tafazolli, Rahim and Bennis, Mehdi and Niyato, Dusit},
  journal={IEEE Wireless Commun.},
  title={Token communications: A large model-driven framework for cross-modal context-aware semantic communications},
  year={2025},
  month={Oct.},
  volume={32},
  number={5},
  pages={80--88},
  doi={10.1109/MWC.001.2500084}
}

@ARTICLE{Meng2024UAVISAC,
  author={Meng, Kaitao and Wu, Qingqing and Xu, Jie and Chen, Wen and Feng, Zhiyong and Schober, Robert and Swindlehurst, A. Lee},
  journal={IEEE Wireless Commun.},
  title={{UAV}-Enabled Integrated Sensing and Communication: Opportunities and Challenges},
  year={2024},
  month={Apr.},
  volume={31},
  number={2},
  pages={97--104},
  doi={10.1109/MWC.131.2200442}
}

@ARTICLE{Zeng2019UAVTutorial,
  author={Zeng, Yong and Wu, Qingqing and Zhang, Rui},
  journal={Proc. IEEE},
  title={Accessing From the Sky: A Tutorial on {UAV} Communications for {5G} and Beyond},
  year={2019},
  month={Dec.},
  volume={107},
  number={12},
  pages={2327--2375},
  doi={10.1109/JPROC.2019.2952892}
}

@ARTICLE{Zhu2024CommSurvey,
  author={Zhu, Changxi and Dastani, Mehdi and Wang, Shihan},
  journal={Auton. Agents Multi-Agent Syst.},
  title={A Survey of Multi-Agent Deep Reinforcement Learning With Communication},
  year={2024},
  month={Jun.},
  volume={38},
  number={1},
  doi={10.1007/s10458-023-09633-6},
  note={Art. no. 4}
}

@ARTICLE{Yates2021AoISurvey,
  author={Yates, Roy D. and Sun, Yin and Brown, III, D. Richard and Kaul, Sanjit K. and Modiano, Eytan and Ulukus, Sennur},
  journal={IEEE J. Sel. Areas Commun.},
  title={Age of Information: An Introduction and Survey},
  year={2021},
  month={May},
  volume={39},
  number={5},
  pages={1183--1210},
  doi={10.1109/JSAC.2021.3065072}
}

@ARTICLE{wei2025token,
  author={Wei, Hao and Ni, Wanli and Wang, Wen and Xu, Wenjun and Niyato, Dusit and Zhang, Ping},
  journal={IEEE Wireless Commun. Lett.},
  title={Token communication in the era of large models: An information bottleneck-based approach},
  year={2026},
  volume={15},
  pages={186--190},
  doi={10.1109/LWC.2025.3622628}
}

@MISC{qiao2025todma,
  author={Qiao, Li and Mashhadi, Mahdi Boloursaz and Gao, Zhen and Schober, Robert and G{\"u}nd{\"u}z, Deniz},
  title={{ToDMA}: Large Model-Driven Massive Token Communications for Semantic Multiple Access},
  year={2025},
  note={arXiv:2505.10946}
}

@ARTICLE{jiang2026agentic,
  author={Jiang, Feibo and Mao, Lei and Dong, Li and Wang, Kezhi and Pan, Cunhua and Jamalipour, Abbas},
  journal={IEEE Trans. Cogn. Commun. Netw.},
  title={Agentic {TokenCom}: A Chain-of-Agents Framework for Multimodal Token Communications},
  year={2026},
  volume={12},
  pages={12051--12064},
  doi={10.1109/TCCN.2026.3732515}
}

@ARTICLE{wu2026toward,
  author={Wu, Jun and Yuan, Weijie and Zhang, Xiaoqi and Yu, Yaohuan and Cui, Yuanhao and Liu, Fan and Sun, Geng and Wang, Jiacheng and Niyato, Dusit and Kim, Dong In},
  journal={IEEE Wireless Commun.},
  title={Toward multi-functional {LAWNs} with {ISAC}: Opportunities, challenges, and the road ahead},
  year={2026},
  doi={10.1109/MWC.2026.3672360},
  note={early access, doi: 10.1109/MWC.2026.3672360}
}

@ARTICLE{wei2024integrated,
  author={Wei, Zhiqing and Liu, Haotian and Li, Hujun and Jiang, Wangjun and Feng, Zhiyong and Wu, Huici and Zhang, Ping},
  journal={IEEE Trans. Mobile Comput.},
  title={Integrated sensing and communication enabled cooperative passive sensing using mobile communication system},
  year={2025},
  month={Sep.},
  volume={24},
  number={9},
  pages={7805--7821},
  doi={10.1109/TMC.2024.3514113}
}

@ARTICLE{yang2025cooperative,
  author={Yang, Yan and Li, Zhendong and Zhao, Jianwei and Wu, Qingqing and Wei, Zhiqing and Chen, Wen and Jia, Weimin},
  journal={IEEE Trans. Wireless Commun.},
  title={Cooperative multi-static {ISAC} networks: A unified design framework for active and passive sensing},
  year={2026},
  volume={25},
  pages={5401--5415},
  doi={10.1109/TWC.2025.3618248}
}

@ARTICLE{wang2026multi,
  author={Wang, Jinyu and Zhang, Xianchao and Wei, Zhiqing and Sun, Fengsong and Li, Yiheng and Feng, Zhiyong and Lu, Jun},
  journal={IEEE J. Sel. Topics Signal Process.},
  title={Multi-{UAV} Collaborative {ISAC} with Dynamic Resource Allocation: A Hierarchical Graph Multi-Agent Reinforcement Learning Approach},
  year={2026},
  doi={10.1109/JSTSP.2026.3693225},
  note={early access, doi: 10.1109/JSTSP.2026.3693225}
}

@MISC{park2026geometric,
  author={Park, Joohyuk and Shin, Junyong and Oh, Yongjeong and Park, Jihong and Jeon, Yo-Seb},
  title={Geometric Cross-Modal Token Selection for Latency-Constrained Multimodal Token Communication},
  year={2026},
  note={arXiv:2608.03141}
}

@MISC{shin2026context,
  author={Shin, Junyong and Park, Joohyuk and Oh, Yongjeong and Park, Jihong and Choi, Jinho and Jeon, Yo-Seb},
  title={Context-Aware Wireless Token Communication via Joint Token Masking and Detection},
  year={2026},
  note={arXiv:2605.02123}
}

@MISC{lee2026hierarchical,
  author={Lee, Jihoon and Oh, Seungeun and Park, Jihong and Kim, Seong-Lyun and Ko, Seung-Woo},
  title={Semantics-Aware Hierarchical Token Communication: Clustering, Bit Mapping, and Power Allocation},
  year={2026},
  note={arXiv:2604.27641}
}

@MISC{zhang2026ada,
  author={Zhang, Zijun and Qiao, Li and Mashhadi, Mahdi Boloursaz and Gao, Zhen and Bennis, Mehdi and Huang, Kaibin},
  title={{Ada-TokenCom}: Rate-Adaptive Token Communications via Large-Model-Driven Token Compression and Generation},
  year={2026},
  note={arXiv:2608.28086}
}

@MISC{liu2026airtf,
  author={Liu, Bole and Qiao, Li and Wu, Minghui and Shao, Yulin and Gao, Zhen},
  title={{AirTF}: Over-the-Air Token Fusion for Task-Oriented Multi-Modal Token Communications},
  year={2026},
  note={arXiv:2607.03099}
}

@ARTICLE{xiang2026aligning,
  author={Xiang, Kun and Zhang, Terry Jingchen and Huang, Yinya and He, Jixi and Liu, Zirong and Tang, Yueling and Zhou, Ruizhe and Luo, Lijing and Wen, Youpeng and Chen, Xiuwei and Lin, Bingqian and Han, Jianhua and Xu, Hang and Li, Hanhui and Dong, Bin and Liang, Xiaodan},
  journal={IEEE Trans. Pattern Anal. Mach. Intell.},
  title={Aligning Perception, Reasoning, Modeling and Interaction: A Survey on Physical {AI}},
  year={2026},
  doi={10.1109/TPAMI.2026.3703038},
  note={early access, doi: 10.1109/TPAMI.2026.3703038}
}

@MISC{niu2024comprehensive,
  author={Niu, Haoyi and Hu, Jianming and Zhou, Guyue and Zhan, Xianyuan},
  title={A comprehensive survey of cross-domain policy transfer for embodied agents},
  year={2024},
  note={arXiv:2402.04580}
}

@INPROCEEDINGS{yu2022surprising,
  author={Yu, Chao and Velu, Akash and Vinitsky, Eugene and Gao, Jiaxuan and Wang, Yu and Bayen, Alexandre and Wu, Yi},
  booktitle={Proc. Adv. Neural Inf. Process. Syst. (NeurIPS)},
  title={The Surprising Effectiveness of {PPO} in Cooperative Multi-Agent Games},
  year={2022},
  volume={35},
  pages={24611--24624}
}

@ARTICLE{luo2022semantic,
  author={Luo, Xuewen and Chen, Hsiao-Hwa and Guo, Qing},
  journal={IEEE Wireless Commun.},
  title={Semantic Communications: Overview, Open Issues, and Future Research Directions},
  year={2022},
  month={Feb.},
  volume={29},
  number={1},
  pages={210--219},
  doi={10.1109/MWC.101.2100269}
}

@ARTICLE{jing2024isac,
  author={Jing, Xiaoye and Liu, Fan and Masouros, Christos and Zeng, Yong},
  journal={IEEE Trans. Wireless Commun.},
  title={{ISAC} from the Sky: {UAV} Trajectory Design for Joint Communication and Target Localization},
  year={2024},
  volume={23},
  number={10},
  pages={12857--12872}
}

@ARTICLE{chen2026differential,
  author={Chen, Ziqi and Du, Jun and Jiang, Chunxiao and Hou, Xiangwang and Han, Zhu and Poor, H. Vincent},
  journal={IEEE Trans. Netw. Sci. Eng.},
  title={Differential privacy-based adaptive sparse federated learning in {UAV} networks},
  year={2026},
  volume={13},
  pages={5128--5144},
  doi={10.1109/TNSE.2025.3644438}
}

@ARTICLE{du2020machine,
  author={Du, Jun and Jiang, Chunxiao and Wang, Jian and Ren, Yong and Debbah, M{\'e}rouane},
  journal={IEEE Veh. Technol. Mag.},
  title={Machine learning for {6G} wireless networks: Carrying forward enhanced bandwidth, massive access, and ultrareliable/low-latency service},
  year={2020},
  month={Dec.},
  volume={15},
  number={4},
  pages={122--134},
  doi={10.1109/MVT.2020.3019650}
}

@ARTICLE{zhang2025task,
  author={Zhang, Hualei and Du, Jun and Jiang, Chunxiao and Wang, Jintao and Bader, Faouzi and Debbah, Merouane},
  journal={IEEE Trans. Veh. Technol.},
  title={Task offloading in {UAV}-assisted mobile cloud-edge computing networks: An {AoP}-aware {HAPPO} approach},
  year={2025},
  month={Sep.},
  volume={74},
  number={9},
  pages={14745--14759},
  doi={10.1109/TVT.2025.3566431}
}

@ARTICLE{yi2026token,
  author={Yi, Peng and Liang, Ying-Chang},
  journal={arXiv preprint arXiv:2610.01826},
  title={Token communication-assisted collaborative embodied artificial intelligence: Concepts, framework, and opportunities},
  year={2026},
  month={Oct.},
  doi={10.48550/arXiv.2610.01826}
}

@ARTICLE{khalili2024efficient,
  author={Khalili, A. and Rezaei, A. and Xu, D. and Dressler, F. and Schober, R.},
  journal={IEEE Trans. Wireless Commun.},
  title={Efficient {UAV} hovering, resource allocation, and trajectory design for {ISAC} with limited backhaul capacity},
  year={2024},
  month={Nov.},
  volume={23},
  number={11},
  pages={17635--17650},
  doi={10.1109/TWC.2024.3455370}
}

\end{document}